\documentclass[12pt,letterpaper]{article}

\usepackage{cite}

\usepackage{array}

\newcolumntype{+}{!{\vrule width 2pt}}

\newlength\savedwidth

\newcommand\thickhline{\noalign{\global\savedwidth\arrayrulewidth\global\arrayrulewidth 2pt}%
\hline
\noalign{\global\arrayrulewidth\savedwidth}}

\usepackage[aboveskip=1pt,labelfont=bf,labelsep=period,singlelinecheck=off]{caption}

\makeatletter
\renewcommand{\@biblabel}[1]{\quad#1.}
\makeatother

\usepackage{amsthm,url,threeparttable,amsmath,amsfonts,graphicx,verbatim}
\usepackage{epstopdf}
\usepackage{authblk}

\newtheorem*{theorem}{Theorem}

\newtheorem{lemma}{Lemma}
\theoremstyle{definition}
\newtheorem*{definition}{Definition}

\newenvironment{case}[1]{\vspace{1ex} \textit{#1}:}

\newcommand{\Prob}{\mathbb{P}}
\newcommand{\E}{\mathbb{E}}
\newcommand{\vx}{\mathbf{x}}

\newcommand{\MSS}{\mathrm{MSS}}

\newcommand{\Ne}{N_\mathrm{eff}}

\begin{document}

\title{Transient amplifiers of selection and reducers of fixation for death-Birth updating on graphs}

\author[1,2]{Benjamin Allen}
\author[1]{Christine Sample}
\author[1]{Robert Jencks}
\author[1]{James Withers}
\author[1]{Patricia Steinhagen}
\author[1]{Lori Brizuela}
\author[1]{Joshua Kolodny}
\author[1]{Darren Parke}
\author[3]{Gabor Lippner}
\author[1]{Yulia A. Dementieva}

\affil[1]{Emmanuel College, Boston, MA, USA}
\affil[2]{Program for Evolutionary Dynamics, Harvard University, Cambridge, MA, USA}
\affil[3]{Department of Mathematics, Northeastern University, Boston, MA, USA}

\date{}

\maketitle

\begin{abstract}
The spatial structure of an evolving population affects which mutations become fixed.  Some structures amplify selection, increasing the likelihood that beneficial mutations become fixed while deleterious mutations do not.  Other structures suppress selection, reducing the effect of fitness differences and increasing the role of random chance. This phenomenon can be modeled by representing spatial structure as a graph, with individuals occupying vertices. Births and deaths occur stochastically, according to a specified update rule.  We study \emph{death-Birth} updating: An individual is chosen to die and then its neighbors compete to reproduce into the vacant spot. Previous numerical experiments suggested that amplifiers of selection for this process are either rare or nonexistent. We introduce a perturbative method for this problem for weak selection regime, meaning that mutations have small fitness effects. We show that fixation probability under weak selection can be calculated in terms of the coalescence times of random walks. This result leads naturally to a new definition of effective population size.  Using this and other methods, we uncover the first known examples of transient amplifiers of selection (graphs that amplify selection for a particular range of fitness values) for the death-Birth process. We also exhibit new families of ``reducers of fixation", which decrease the fixation probability of all mutations, whether beneficial or deleterious. 
\end{abstract}

\section*{Introduction}

Spatial population structure has a variety of effects on natural selection \cite{NowakMay,HauertSnowdrift,Lieberman2005Graphs,Ohtsuki,allen2017evolutionary}.  These effects can be studied mathematically by representing spatial structure as a graph \cite{Lieberman2005Graphs}. The vertices represent individuals, and the edges indicate spatial relationships between them.  This modeling approach, known as evolutionary graph theory, has illuminated the effects of spatial structure on the rate of genetic change \cite{Allen2015molecular}, the balance of selection versus neutral drift \cite{Lieberman2005Graphs,shakarian2012review,pavlogiannis2018construction}, and the evolution of cooperation and other social behaviors \cite{Ohtsuki,Taylor,chen2013sharp,debarre2014social,durrett2014spatial,pena2016evolutionary,allen2017evolutionary,fotouhi2018conjoining}.  

Here we focus on how spatial structure affects fixation probability---the probability that a new mutation will spread throughout the population, depending on its effect on fitness.  Previous work \cite{Lieberman2005Graphs,antal2006evolutionary,Broom2008Analysis,frean2008death,shakarian2012review,voorhees2013fixation,monk2014martingales,adlam2015amplifiers,cuesta2015fast,hindersin2015most,jamieson2015fixation,pavlogiannis2017amplification,cuesta2017suppressors,pavlogiannis2018construction,cuesta2018evolutionary} has shown that some graphs act as \emph{amplifiers} of selection, increasing the fixation probability of beneficial mutations, while reducing that of deleterious mutations.  Other graphs act as \emph{suppressors} of selection, increasing the fixation probability of deleterious mutations and reducing that of beneficial mutations.  Over time, a population that is structured as an amplifier will more rapidly accrue beneficial mutations, whereas one structured as a suppressor will experience greater effects of random drift.

To be precise, the terms amplifier and suppressor cannot be ascribed solely to a graph itself. Fixation probabilities also depend on the \emph{update rule:} the scheme by which births and deaths are determined.  The majority of works on amplifiers and suppressors use \emph{Birth-death (Bd)} updating: An individual is selected to reproduce proportionally to fitness, and its offspring replaces a uniformly-chosen neighbor.  A minority of works \cite{frean2008death,kaveh2015duality,pattni2015evolutionary,hindersin2015most} have considered \emph{death-Birth (dB)} updating: A uniformly-chosen individual dies, and a neighbor is chosen proportionally to fitness to reproduce into the vacancy. (Following Hindersin and Traulsen \cite{hindersin2015most}, we use uppercase letters for a demographic step that is affected by fitness, and lowercase letters for a step that is fitness-independent.)  Interestingly, the choice of update rule has a marked effect on fixation probabilities.  For example, the Star graph (Fig.~\ref{fig:completestarcycle}B) is an amplifier of selection for Bd updating \cite{Lieberman2005Graphs,Broom2008Analysis}, but a suppressor for dB updating \cite{frean2008death}.  

\begin{figure}
\begin{center}
\includegraphics[scale=0.8]{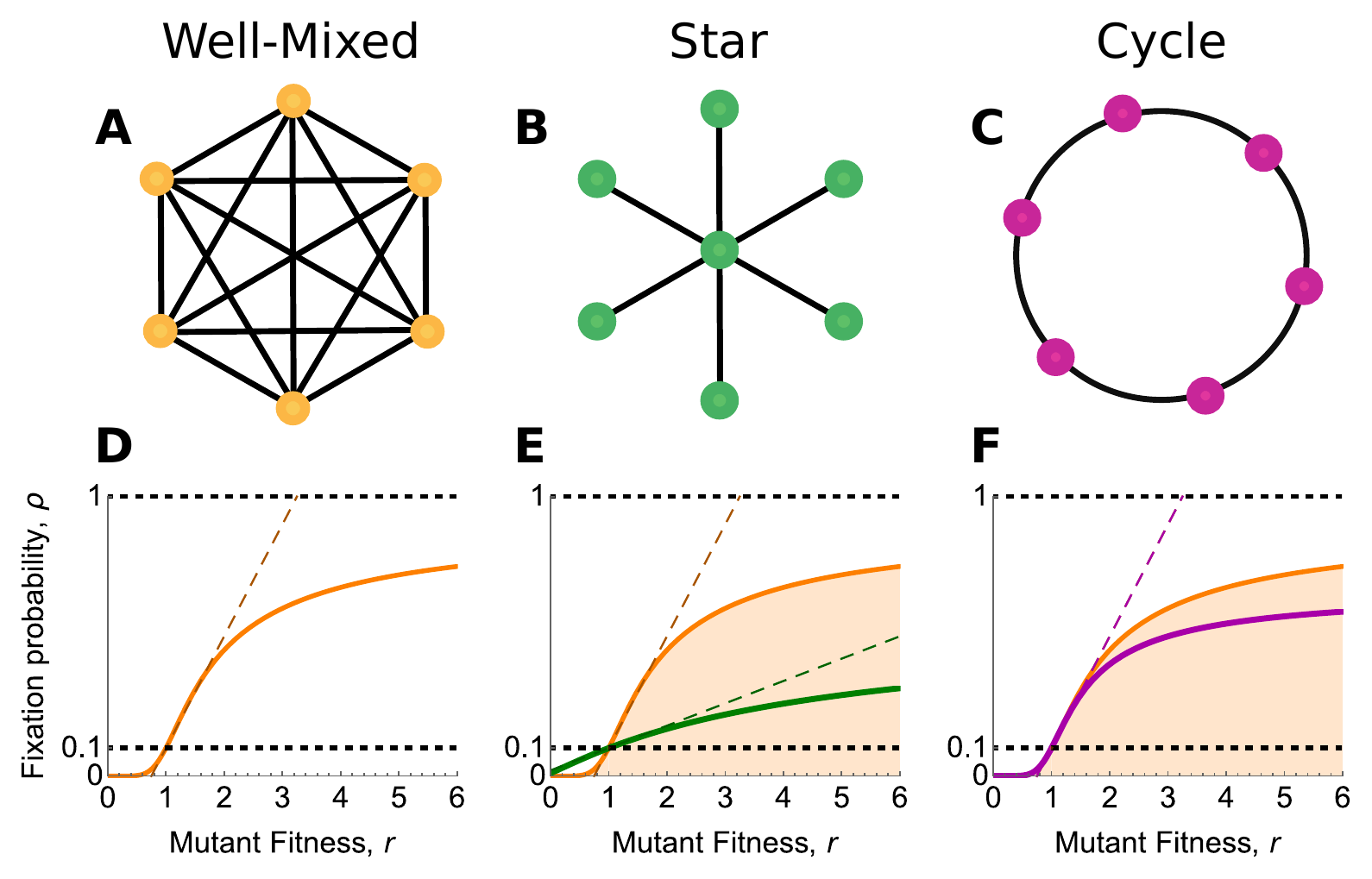}
\caption{\textbf{Fixation probabilities for constant selection on graphs.} \textbf{(A}) The complete graph $K_N$ represents a well-mixed population. \textbf{(B)} The star $S_n$ consists of one hub vertex connected to $n$ leaf vertices.  This star is a suppressor of selection for death-Birth updating \cite{frean2008death}. \textbf{(C)} The cycle $C_N$, a regular graph of degree 2, is a reducer of fixation: the fixation probability of any mutant type of fitness $r\neq 1$ is smaller than it would be in the well-mixed case \cite{hindersin2015most}. Panels \textbf{(D)--(F)} plot fixation probability versus mutant fitness for the respective graphs, for size $N=10$, with the well-mixed case (orange curve) shown for comparison. Dashed lines show the linear approximation to fixation probability at $r=1$.  These approximations are accurate for weak selection ($r \approx 1$) and can be computed from coalescence times using Eqs.~\eqref{eq:trecur}--\eqref{eq:Ne}.}
\label{fig:completestarcycle}
\end{center}
\end{figure}

A recent numerical investigation \cite{hindersin2015most} of thousands of random graphs up to size 14 found no amplifiers of selection for death-Birth updating.  This suggests that amplifiers for dB are either nonexistent or rare, at least among small graphs.  This work also identified a graph (the cycle; Fig.~\ref{fig:completestarcycle}C) that, for dB updating, reduces fixation probabilities for \emph{all} mutations that affect fitness, whether beneficial or deleterious. The cycle is therefore neither an amplifier nor a suppressor; it might instead be called a ``reducer of fixation", in that it preserves the resident wild-type regardless of fitness effects.  A follow-up work \cite{hindersin2016should} identified other reducers of fixation.

Here we investigate fixation probabilities for death-Birth updating on graphs, using a variety of analytical and numerical methods.  We develop a weak-selection approach to this question, based on coalescing random walk methods \cite{cox1989coalescing,liggett2006interacting} that were previously used to study evolutionary games on graphs \cite{chen2013sharp,allen2017evolutionary,fotouhi2018conjoining}.  Weak selection means that the fitness of the mutant is close to that of the resident; i.e., the mutation is either slightly beneficial or slightly deleterious.  Unlike earlier numerical methods \cite{hindersin2015most,hindersin2016exact,cuesta2017suppressors}, the weak-selection method can be performed in polynomial time, allowing for efficient identification of amplifiers and suppressors of weak selection.  We apply this method to several graph families and random graph models.  We also compute fixation probabilities for arbitrary mutant fitness (beyond weak selection) for these graph families.

We find, contrary to the expectation set by previous numerical experiments \cite{hindersin2015most}, that amplifiers, of a sort, do exist for death-Birth updating.  Specifically, we exhibit several families of \emph{transient amplifiers}, which amplify selection only for a certain range of mutant fitness values. We also uncover new examples of reducers of fixation. 

Our weak-selection method also leads to new theoretical results. First, the form of our expression for fixation probability suggests a new definition of effective population size, with intriguing connections to previous definitions \cite{wright1931evolution,kimura1963measurement,felsenstein1971inbreeding,crow1988inbreeding,antal2006evolutionary,broom2012two,allen2013adaptive,giaimo2018invasion}. Second, we show that for \emph{isothermal} graphs---which have the same edge weight sum at each vertex---the fixation probability coincides, under weak selection, with that of a well-mixed population.  This result is reminiscent of the Isothermal Theorem of Lieberman et al.~\cite{Lieberman2005Graphs}, which applies to Bd updating (see also Refs.~\cite{kaveh2015duality,pattni2015evolutionary}).  However, whereas the original Isothermal Theorem is valid for any strength of selection, our new result applies only to weak selection. Third, we show that fixation probabilities under weak selection can be well-approximated using only the first two moments of the degree distribution. This approximation helps explain why amplifiers of selection (even transient ones) are rare for dB updating.

\section*{Model}

We study an established model of natural selection on graphs \cite{Lieberman2005Graphs,antal2006evolutionary,Broom2008Analysis,frean2008death,shakarian2012review,voorhees2013fixation,monk2014martingales,adlam2015amplifiers,cuesta2015fast,hindersin2015most,jamieson2015fixation,kaveh2015duality,pattni2015evolutionary,pavlogiannis2017amplification,cuesta2017suppressors,hathcock2018fitness,pavlogiannis2018construction,cuesta2018evolutionary,moller2019exploring,tkadlec2019population}. Spatial structure is represented as a connected, weighted, undirected graph $G$.  Joining each pair of vertices $i$ and $j$ is an edge of weight $w_{ij} \geq 0$, with $w_{ij}=w_{ji}$ since $G$ is undirected.  We exclude the possibility of self-loops by setting $w_{ii}=0$ for each vertex $i$.  The size of the graph, which is also the population size, is denoted $N$.

Each vertex houses a single haploid individual.  Individuals can be of mutant or resident (wild-) type.  Mutants have fitness $r>0$, while the fitness of the resident type is set to 1.  Advantageous mutants have $r>1$, while deleterious mutants have $r<1$.  The case $r=1$ describes neutral drift, for which the mutation has no fitness effect.  This model describes \emph{constant selection}, in that the fitnesses of the competing types do not vary with the current population state.

Selection proceeds according to the death-Birth (dB) update rule \cite{Ohtsuki,frean2008death,allen2014games}.  First, an individual is selected uniformly at random for death, creating a vacant vertex.  Then, a neighbor of the vacant vertex is chosen to reproduce, with probability proportional to (fitness) $\times$ (edge weight to the vacant vertex).  The new offspring fills the vacancy, inheriting the type of the parent.

As an initial state, we suppose that a single mutant is introduced, at a vertex chosen uniformly at random, in a population otherwise composed of residents.  We define the mutation's \emph{fixation probability} as the expected probability that a state of all mutants is reached from this initial condition.  The fixation probability of a mutation of fitness $r$ on a graph $G$ is denoted $\rho_G(r)$.  

The baseline case of a well-mixed population is represented by the complete graph $K_N$ of size $N$ (Fig \ref{fig:completestarcycle}A).  For dB updating on $K_N$, a mutant of fitness $r$ has fixation probability \cite{kaveh2015duality,hindersin2015most}
\begin{equation}
\label{eq:WM}
\rho_{K_N}(r) = \frac{N-1}{N} \frac{1-r^{-1}}{1-r^{-(N-1)}}.
\end{equation}
We characterize the effects of graph structure on fixation probabilities using the following definitions:
\begin{definition}
Let $G$ be a graph of size $N$.  Then $G$ is
\begin{itemize}
    \item An \emph{amplifier of selection} if $\rho_G(r) < \rho_{K_N}(r)$ for $0<r<1$ and $\rho_G(r) > \rho_{K_N}(r)$ for $r>1$.
    \item A \emph{suppressor of selection} if $\rho_G(r) > \rho_{K_N}(r)$ for $0<r<1$ and $\rho_G(r) < \rho_{K_N}(r)$ for $r>1$.
    \item A \emph{transient amplifier of selection} if there is some $r^*>1$ such that $\rho_G(r) < \rho_{K_N}(r)$ for  $0<r<1$ and for $r>r^*$, and $\rho_G(r) > \rho_{K_N}(r)$ for $1<r<r^*$.
    \item A \emph{reducer of fixation} if $\rho_G(r) < \rho_{K_N}(r)$ for all $r \neq 1$.
\end{itemize}
\end{definition}
For example, the star graph $S_n$ with $n$ leaves (population size $N=n+1$; Fig.~\ref{fig:completestarcycle}B) is a suppressor of selection for dB updating \cite{frean2008death}, with fixation probability \cite{hadjichrysanthou2011evolutionary}
\begin{equation}
\label{eq:starfix}
\rho_{S_n}(r)=\frac{(N-1)r+1}{N(r+1)} \left( \frac{1}{N} + \frac{r}{N+2r-2} \right).
\end{equation}
The cycle $C_N$ is a reducer of fixation for dB updating \cite{hindersin2015most}, with fixation probability \cite{kaveh2015duality}
\begin{equation}
\label{eq:cyclefix}
\rho_{C_N}(r) = \frac{2(r-1)}{3r-1+r^{-(N-1)}-3r^{-(N-2)}}.
\end{equation}
Other examples of reducers were identified by Hindersin et al.~\cite{hindersin2016should}, who called them ``suppressors of evolution"; we prefer ``reducers of fixation" to avoid confusion with suppressors of selection. 

A companion work \cite{Pepa} proves that there are no (non-transient) amplifiers of selection for dB updating. Transient amplifiers of selection were previously known for Bd updating \cite{voorhees2013fixation} but not for dB updating. For Bd updating, there are some graphs that do not fit any of the above definitions, but alternate between amplification and suppression (i.e,~$\rho_G(r)>\rho_{K_N}(r)$ on a disconnected set of $r$-values) \cite{cuesta2018evolutionary}; such examples have not been discovered for dB updating.

\section*{Results}

\subsection*{Fixation probability under weak selection}

Fixation probabilities on graphs can be difficult to compute.  Current numerical methods \cite{cuesta2015fast,hindersin2015most,hindersin2016exact,cuesta2017suppressors} involve solving a system of $\mathcal{O}(2^N)$ equations to compute fixation probabilities on a given graph of size $N$.  For this reason, previous analyses have focused on small graphs \cite{hindersin2015most,hindersin2016exact,cuesta2017suppressors,cuesta2018evolutionary,moller2019exploring,tkadlec2019population} and/or graphs with a high degree of symmetry \cite{Lieberman2005Graphs,Broom2008Analysis,shakarian2012review,voorhees2013fixation,monk2014martingales,adlam2015amplifiers,jamieson2015fixation,pavlogiannis2017amplification}.

One way to mitigate these difficulties is to focus on \emph{weak selection}, which is the regime $r \approx 1$.  Weak selection can be studied as a perturbation of neutral drift ($r=1$). This approach has been fruitfully applied to population genetics \cite{Haldane1924mathematical,kimura1962probability,akashi2012weak} and evolutionary game theory \cite{NowakFinite,Ohtsuki,chen2013sharp,allen2014games,debarre2014social,allen2017evolutionary,fotouhi2018conjoining}, but so far has not been applied to models of constant selection on graphs.

To implement weak selection for our model, we write the fitness of the mutant as $r=1+\delta$, with $\delta$ representing the mutation's selection coefficient. We consider the first-order Taylor expansion of the fixation probability, $\rho_G(1+\delta)$, at $\delta=0$.  For the complete graph, Taylor expansion of Eq.~\eqref{eq:WM} yields
\begin{equation}
\label{eq:WMexpand}
\rho_{K_N}(1 + \delta) = \frac{1}{N} + \delta \frac{N-2}{2N} + \mathcal{O}(\delta^2).
\end{equation}

For an arbitrary weighted, connected graph, we apply a method developed by Allen et al.~\cite{allen2017evolutionary} to calculate fixation probabilities under weak selection. This method uses \emph{coalescing random walks}, which trace the co-ancestry of given individuals backwards in time to their most recent common ancestor.

Each individual's ancestry is represented as a random walk on $G$.  These random walks are defined by the step probabilities $p_{ij} = w_{ij}/w_i$, where $w_i = \sum_{j \in G} w_{ij}$ is the \emph{weighted degree} of vertex $i$. Importantly, $p_{ij}$ is also equal to the conditional probability, under neutral drift ($r=1$), that $j$ reproduces, given that $i$ is replaced.  Random walks on $G$ have a stationary distribution, in which the probability of being at vertex $i$ is equal to its relative weighted degree, $\pi_i=w_i/\left(\sum_{j \in G} w_j \right)$.

To represent the co-ancestry of two individuals, we consider a pair of random walkers.  At each time-step, one of the two walkers is chosen (with equal probability) to take a step.  The point at which the two walkers meet (coalesce) represents the most recent common ancestor.  We let $\tau_{ij}$ denote the expected time to coalescence from initial vertices $i$ and $j$.  These coalescence times can be determined from the following system of equations \cite{allen2017evolutionary,allen2018mathematical}:
\begin{equation}
\label{eq:trecur}
\tau_{ij} = \begin{cases} 0 & i=j\\
1 + \frac{1}{2}  \sum_{k \in G} \left(p_{ik} \tau_{jk} + p_{jk} \tau_{ik} \right)& i \neq j.
\end{cases}
\end{equation}

We also define the \emph{remeeting time} $\tau_i$ from vertex $i$ as the expected time for two random walkers from vertex $i$ to rejoin each other.  Remeeting times are related to coalescence times by
\begin{equation} 
\label{eq:ti}
\tau_i=1+\sum_{j\in G}p_{ij}\tau_{ij},
\end{equation}
and obey the identity \cite{allen2017evolutionary}
\begin{equation}
\label{eq:pi2tau}
\sum_{i \in G} \pi_i^2 \tau_i = 1.
\end{equation}

Applying the properties of coalescence times, we prove in Appendix \ref{sec:weak} that fixation probability on an arbitrary (weighted, undirected, connected) graph $G$ can be expanded under weak selection as
\begin{equation}
\label{eq:rhoexpand}
\rho_G(1 + \delta) = \frac{1}{N} + \delta \frac{\Ne-2}{2N} + \mathcal{O}(\delta^2),
\end{equation}
where $\Ne$ is the \emph{effective population size} of $G$, which we define as
\begin{equation}
\label{eq:Ne}
\Ne=\sum_{i \in G} \pi_i \tau_i.
\end{equation}
This definition of effective population size is distinct from, but closely related to, previous definitions \cite{wright1931evolution,kimura1963measurement,felsenstein1971inbreeding,crow1988inbreeding,antal2006evolutionary,broom2012two,allen2013adaptive,giaimo2018invasion}, as we review in the Discussion.

Comparing the first-order terms in Eqs.~\eqref{eq:rhoexpand} and \eqref{eq:WMexpand} provides a criterion for the effects of graph structure on fixation probabilities under weak selection:
\begin{definition}
Let $G$ be a graph of size $N$. We say $G$ is 
\begin{itemize}
    \item An \emph{amplifier of weak selection} if $\Ne>N$,
    \item A \emph{suppressor of weak selection} if $\Ne<N$.
\end{itemize}
\end{definition}
An amplifier (respectively, suppressor) of weak selection is guaranteed to amplify (respectively, suppress) selection for all $r$ sufficiently close to 1. Formally, if $G$ is an amplifier of weak selection, there exist $a,b$ with $0 \leq a<1<b \leq \infty$ such that $\rho_G(r) < \rho_{K_N}(r)$ for $a<r<1$ and $\rho_G(r) > \rho_{K_N}(r)$ for $1<r<b$.  Likewise, if $G$ is a suppressor of weak selection, there exist $a,b$ with $0 \leq a<1<b \leq \infty$ such that $\rho_G(r) > \rho_{K_N}(r)$ for $a<r<1$ and $\rho_G(r) < \rho_{K_N}(r)$ for $1<r<b$.  

As an example, solving Eq.~\eqref{eq:trecur} for the star graph $S_n$, and applying Eqs.~\eqref{eq:ti} and \eqref{eq:Ne}, we obtain $\tau_H=\tau_L=\Ne=4n/(n+1)$.  Since the star graph has size $N=n+1$, we find that the star is a suppressor of weak selection for all $n\geq2$.  Substituting in Eq.~\eqref{eq:rhoexpand}, we obtain
\begin{equation}
\rho_G(1 + \delta) = \frac{1}{N} + \delta \frac{N-2}{N^2} + \mathcal{O}(\delta^2),
\end{equation}
which agrees with the Taylor expansion of Eq.~\eqref{eq:starfix}.

\subsection*{Weak-selection Isothermal Theorem}

A particularly interesting result arises in the special case of \emph{isothermal} graphs.  An unweighted graph $G$ is isothermal if each vertex has the same weighted degree $w_i$, or equivalently, if $\pi_i=1/N$ for each $i\in G$. The \emph{Isothermal Theorem} \cite{Lieberman2005Graphs} states that, for Bd updating, an isothermal graph has the same fixation probabilities as a well-mixed population of the same size, for all values of $r$ and all starting configurations of mutants.  However, the corresponding statement for dB updating is false \cite{kaveh2015duality,pattni2015evolutionary}.  For example, the cycle (Fig.~\ref{fig:completestarcycle}C) is isothermal, but its fixation probabilities, as given by Eq.~\eqref{eq:cyclefix}, differ from those of a well-mixed population, given by Eq.~\eqref{eq:WM}.

Here we show that a weak-selection version of the isothermal theorem holds for death-Birth updating.  For an isothermal graph $G$, Eqs.~\eqref{eq:pi2tau} and \eqref{eq:Ne} give
\begin{equation}
\Ne = \sum_{i\in G} \left( \frac{1}{N}\right)\tau_i
= N \sum_{i\in G} \left( \frac{1}{N^2}\right)\tau_i
= N \sum_{i\in G} \pi_i^2 \tau_i = N.
\end{equation}
Combining with Eq.~\eqref{eq:rhoexpand}, we arrive at the following result:

\begin{theorem}[Weak-selection Isothermal Theorem for dB Updating]
Let $G$ be a weighted, undirected, connected isothermal graph of size $N \geq 2$ with no self-loops.  Then for dB updating, fixation probabilities on $G$ coincide with those on the complete graph $K_N$ to first order in the selection coefficient $\delta$:
\begin{equation}
\rho_G(1+\delta) = \rho_{K_N}(1+\delta) + \mathcal{O}(\delta^2).
\end{equation}
\end{theorem}

In other words, if $G$ is isothermal, then the plot of $\rho_G(r)$ is tangent to that of $\rho_{K_N}(r)$  at $r=1$.  This implies that isothermal graphs are neither amplifiers nor suppressors of weak selection.  For example, the cycle $C_N$ (Fig.~\ref{fig:completestarcycle}C),  is isothermal, and therefore the plots of $\rho_{C_N}(r)$ and $\rho_{K_N}(r)$ are tangent at $r=1$ (Fig.~\ref{fig:completestarcycle}F).  However, these plots do not coincide beyond $r=1$; instead, $\rho_{C_N}(r)<\rho_{K_N}(r)$ for all $r\neq 1$ \cite{hindersin2016exact}, meaning that the cycle is a reducer of fixation.

\subsection*{Examples}

We now introduce three example families of graphs, which can behave as transient amplifiers, suppressors, or reducers, depending on the parameter values.  We analyze these graphs both for weak and nonweak selection. Our results are summarized in Table \ref{tab:results}; detailed calculations and proofs are presented in Appendix \ref{sec:examples}.

\begin{table}[!ht]
    \centering
    \begin{threeparttable}
    \caption{\bf{Results for example graphs}}
    \label{tab:results}
    \begin{tabular}{|l+l|l|l|}
    \hline
    \textbf{Example} & \textbf{Case} & \textbf{Classification}\\
    \thickhline
    Separated Hubs\tnote{*} & $n \leq h$ & Suppressor\\
     ($\epsilon \to 0$)   & $n=h+1$ & Reducer\\
        & $n \geq h+2$ & Transient amplifier\\
    \hline
    Star of Islands & $m \leq h-1$ & Suppressor\tnote{**}\\
     ($\epsilon \to 0$)   & $m=h$ & Reducer\\
        & $m \geq h+1$ & Transient amplifier\tnote{**}\\
        \hline
    \end{tabular}
    \begin{tablenotes}
    \item[*] The Fan is the $h=1$ case of separated hubs.
    \item[**] Proven only for weak selection (other cases are proven for arbitrary selection strength).
    \end{tablenotes}
\end{threeparttable}
\end{table}

\subsubsection*{Fan}

\begin{figure}
    \centering
    \includegraphics[scale=0.6]{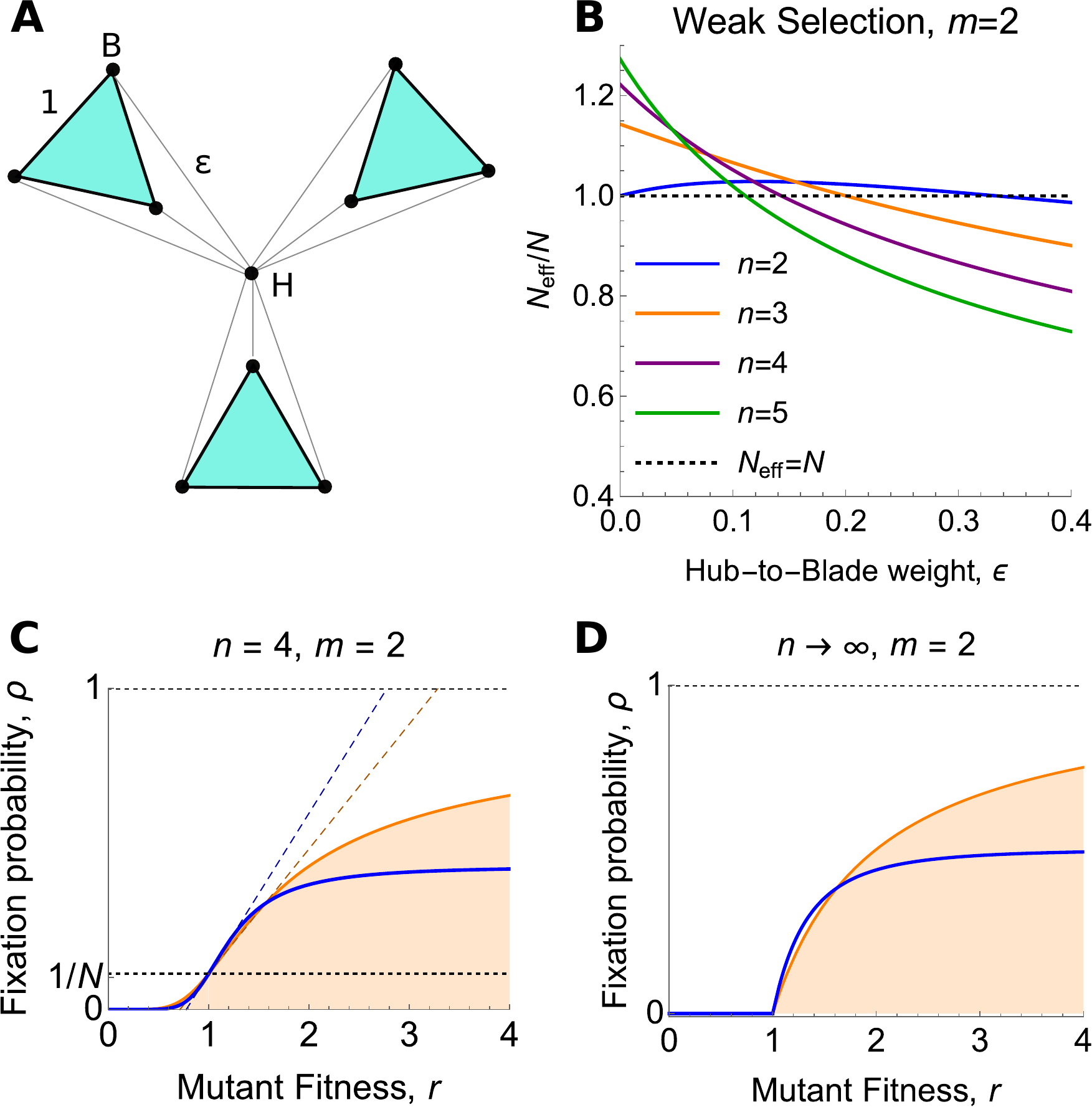}
    \caption{\textbf{The Fan} \textbf{(A)} The Fan, $F_{n,m}$, consists of one hub and $n \geq 2$ ``blades", with $m \geq 2$ vertices per blade. Edge weights are as shown. The case $n=m=3$ is pictured. \textbf{(B)} The ratio of effective versus actual population size, plotted against the hub-to-blade edge weight $\epsilon$, for $m=2$ vertices per blade. For $n=2$ blades, the Fan is an amplifier of weak selection for $0<\epsilon<1/3$, but becomes a reducer in the $\epsilon \to 0$ limit. For $n \geq 3$, the Fan is a transient amplifier for sufficiently small $\epsilon$, including the $\epsilon \to 0$ limit. \textbf{(C)} Fixation probability for $F_{4,2}$ (blue curve), plotted against mutant fitness $r$, in the $\epsilon \to 0$ limit, according to Eq.~\eqref{eq:Fan_fixprob}. The orange curve shows the corresponding well-mixed population result, Eq.~\eqref{eq:WM}, for comparison. Dotted lines show the corresponding weak-selection results (i.e. the linear approximation at $r=1$), according to Eqs.~\eqref{eq:WMexpand},\eqref{eq:rhoexpand}, and \eqref{eq:NeFan}. \textbf{(D)} In the $n \to \infty$ limit, fixation probability is given by Eq.~\eqref{eq:Fan_fixprob}, and the Fan is an amplifier for $1<r<(1+\sqrt{5})/2$. }
    \label{fig:Fan}
\end{figure}

The Fan, $F_{n,m}$, (Fig.~\ref{fig:Fan}) has one hub and $n \geq 2$ blades.  Each blade contains $m \geq 2$ vertices, for a total of $N=nm+1$ vertices.  Each blade vertex is joined to the hub by an edge of weight $\epsilon>0$, and is joined to each other vertex on the same blade by an edge of weight 1. The Fan is isothermal when $\epsilon=(m-1)/(nm-1)$.

Applying our weak-selection method, we find that the Fan has effective population size
\begin{equation}
\label{eq:NeFan}
\Ne=\frac{m n (m n+4 \epsilon -1) (m (m n+\epsilon -1)-\epsilon )}{(m-1) \epsilon ^2 (m n+1)+\epsilon  (m (n+2)-1) (m n-1)+m (m n-1)^2}.
\end{equation}
Comparing to $N=nm+1$, we find that the Fan amplifies weak selection for all $0<\epsilon<(m-1)/(nm-1)$ (Fig.~\ref{fig:Fan}B).  An interesting behavior occurs in the limit $\epsilon \to 0$: For $n\ge 3$ blades the Fan amplifies weak selection ($\lim_{\epsilon \to 0} \Ne > N$), but for $n=2$ blades the Fan neither amplifies nor suppresses weak selection ($\lim_{\epsilon \to 0}\Ne=N$).  The strongest amplifier of weak selection (largest $\Ne/N$) occurs for $m=2$ and first $\epsilon \to 0$ and then $n \to \infty$; in this case, $\Ne/N \to 3/2$.

Moving beyond weak selection, we calculate the fixation probability for a mutation of arbitrary fitness $r>0$, in the $\epsilon \to 0$ limit:
\begin{equation}
\label{eq:Fan_fixprob}
\rho_{F_{n,m}}(r)=\frac{n(m-1) \left(1-r^{-1} \right) \left(1-r^{-(m+1)} \right)}
{( mn+1) \left(1-r^{-(m-1)}\right)\left(1-r^{-n(m+1)}\right)}.
\end{equation}
We prove in Appendix \ref{sect:fan} that, in the $\epsilon \to 0$ limit, the Fan is a reducer of fixation for $n=2$ and a transient amplifier of selection for all $n \geq 3$.

\subsubsection*{Separated Hubs}

\begin{figure}
    \centering
    \includegraphics[scale=0.6]{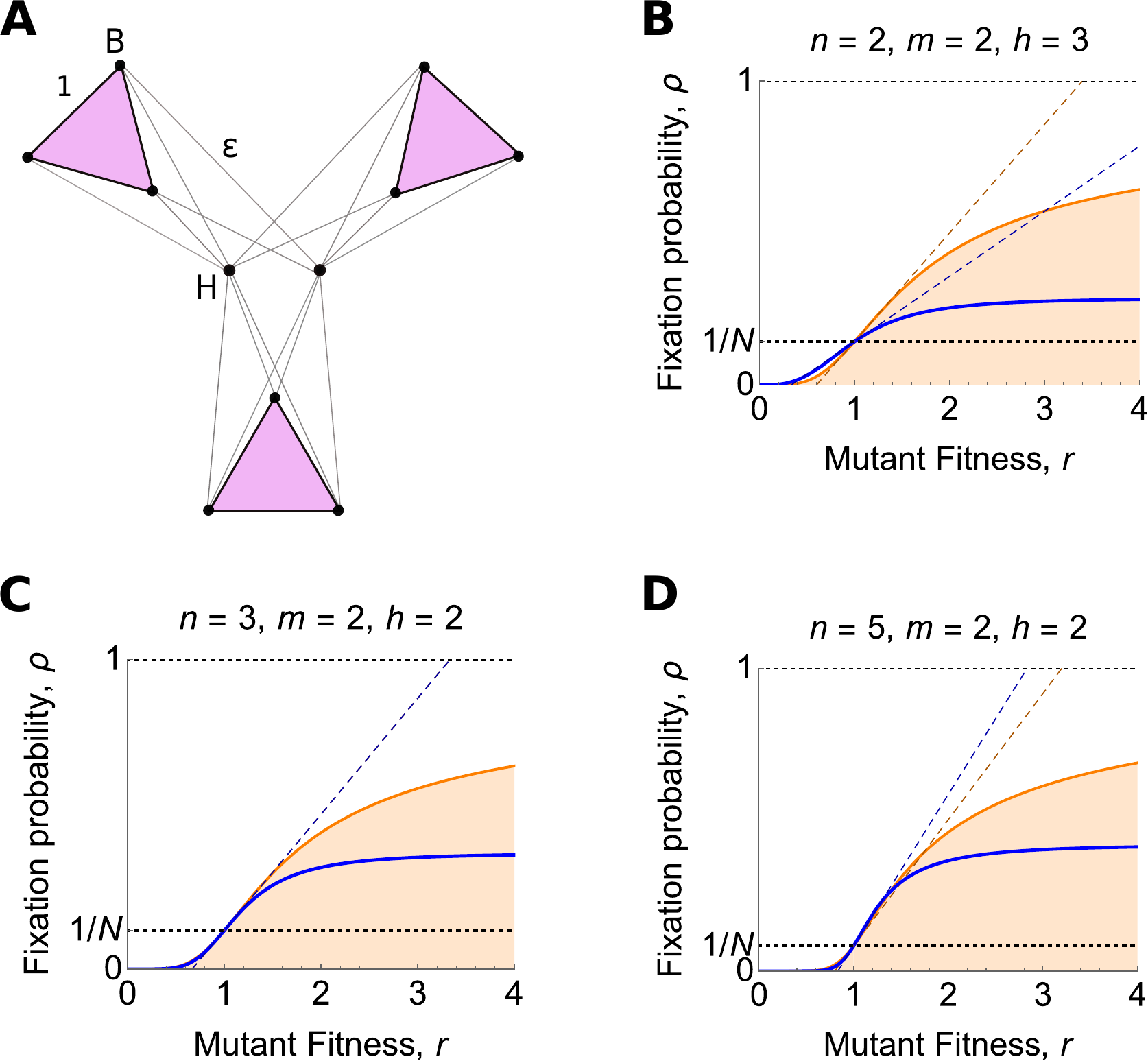}
    \caption{\textbf{Separated Hubs} \textbf{(A)} The Separated Hubs graph consists of $h \geq 1$ hubs and $n \geq 1$ blades, with $m \geq 2$ vertices per blade. Edge weights are as shown. \textbf{(B)--(D)} Blue curves show fixation probability, Eq.~\eqref{eq:SH_fixprob}, plotted against mutant fitness $r$, in the $\epsilon \to 0$ limit. Blue dotted lines show the weak selection result, Eqs.~\eqref{eq:rhoexpand} and \eqref{eq:SH_Ne}. The orange curve and dotted line show the corresponding well-mixed population results, Eqs.~\eqref{eq:WM} and \eqref{eq:WMexpand}, for comparison. The Separated Hubs graph is (B) a suppressor for $n \le h$, (C) a reducer for $n=h+1$, and (D) a transient amplifier for $n\ge h+2$.}
    \label{fig:SH}
\end{figure}

Our next examples generalize the Fan graph in two different ways. First, we suppose that there are multiple hub vertices, which are not connected to each other.  The resulting graph, which we call the \emph{Separated Hubs} graph, $SH_{n,m,h}$, has $h\geq 1$ hub vertices, $n\geq 2$ blades, and $m \geq 2$ vertices per blade (Fig.~\ref{fig:SH}).  Vertices on the same blade are connected by edges of weight 1, and each blade vertex is connected to each hub by an edge of weight $\epsilon$. No other edges are present. The Fan is the $h=1$ case of Separated Hubs.

The weak-selection results for arbitrary $\epsilon$ are rather cumbersome, but in the $\epsilon \to 0$ limit they simplify to
\begin{equation}
\label{eq:SH_Ne}
\Ne=nm + n - 1.
\end{equation}
Interestingly, in this limit, the effective population size is independent of the number $h$ of hubs. Comparing Eq.~\eqref{eq:SH_Ne} to the population size, $N=nm+h$, we observe that the Separated Hubs graph (in the $\epsilon \to 0$ limit) is a suppressor of weak selection for $n \leq h$ and an amplifier of weak selection for $n\geq h+2$. Again, the strongest amplifier of weak selection occurs for $m=2$ and first $\epsilon \to 0$ and then $n \to \infty$, leading to $\Ne/N \to 3/2$.  The strongest suppressor of weak selection (smallest $\Ne/N$) occurs for first $\epsilon \to 0$ and then $h\to \infty$, leading to $\Ne/N \to 0$.

Beyond weak selection, we compute the fixation probability for arbitrary $r>0$ in the limit $\epsilon \to 0$:
\begin{align}
\label{eq:SH_fixprob}
\rho_{SH_{n,m,h}}(r)=\frac{n(m-1)(1-r^{-1})\left(1-r^{-(m+1)}\right)}
{( mn+h)\left(1-r^{-(m-1)}\right)\left(1-r^{-n(m+1)}\right)}.
\end{align}
In the limit of many blades, we obtain
\begin{equation}
\lim_{n \to \infty} \rho_{SH_{n,m,h}}(r) =
\begin{cases}
0 & 0\leq r\leq 1\\
\frac{m-1}{m} \frac{(1-r^{-1})(1-r^{-(m+1)})}{1-r^{-(m-1)}}
 &r>1.
\end{cases}
\end{equation}

We prove in Appendix \ref{app:SI} that the Separated Hubs graph, in the $\epsilon \to 0$ limit, is a suppressor for $n \le h$, a transient amplifier for $n \ge h+2$, and a reducer for $n=h+1$.

\subsubsection*{Star of Islands}

\begin{figure}
    \centering
    \includegraphics[scale=0.6]{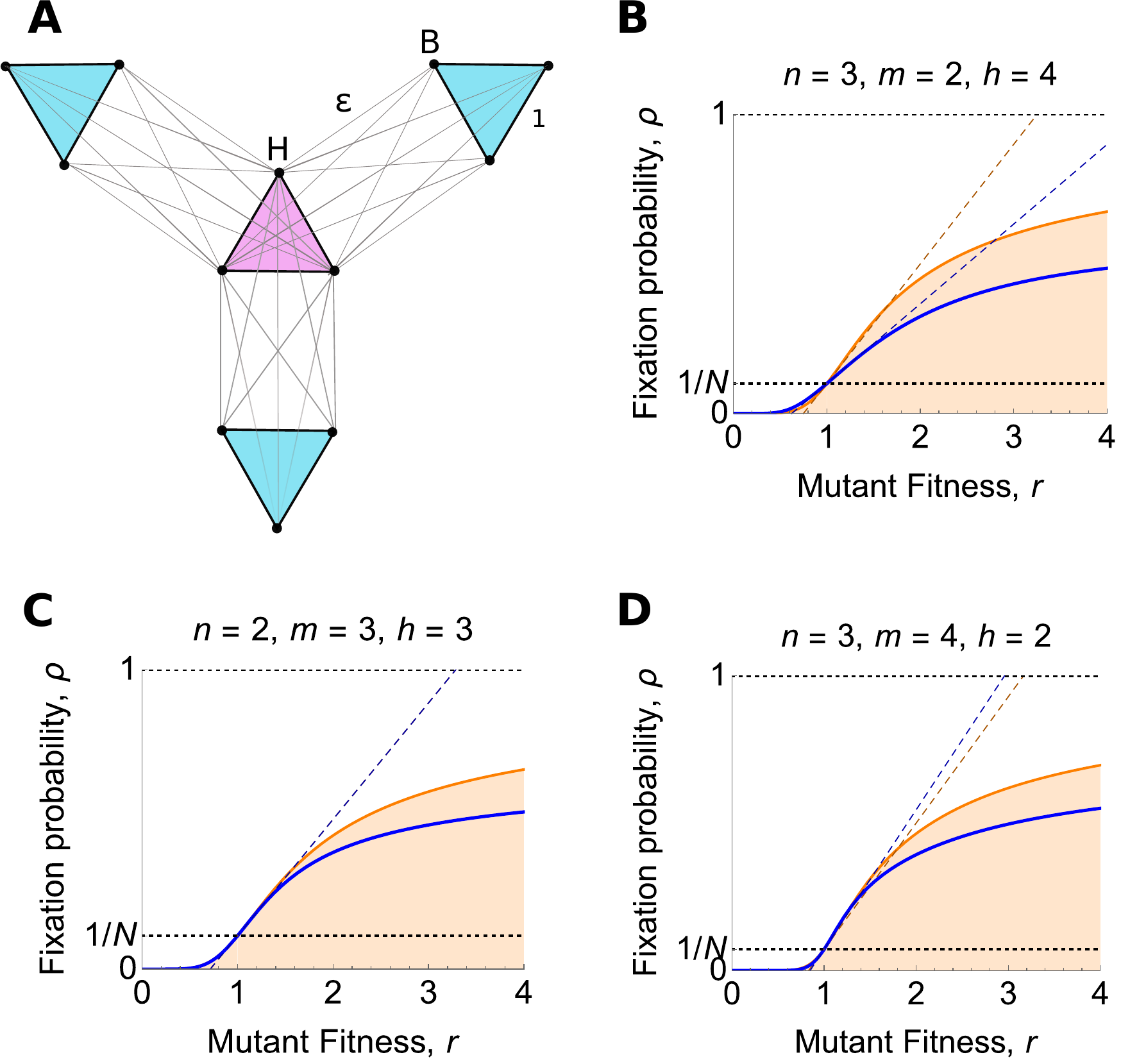}
        \caption{\textbf{Star of Islands} \textbf{(A)} The Star of Islands graph consists of a hub island of size $h \geq 2$, and $n \geq 1$ other islands of size $m \geq 2$. Edge weights are as shown. \textbf{(B)--(D)} Blue curves show fixation probability, Eqs.~\eqref{eq:SI_fixprob2}--\eqref{eq:SI_fixprob4}, plotted against mutant fitness $r$, in the $\epsilon \to 0$ limit. Blue dotted lines show the weak selection result, Eqs.~\eqref{eq:rhoexpand} and \eqref{eq:SI_Ne}. The orange curve and dotted line show the corresponding well-mixed population results, Eqs.~\eqref{eq:WM} and \eqref{eq:WMexpand}, for comparison. The Star of Islands graph is (B) a suppressor for $m\ge h-1$, (C) a reducer for $m=h$, and (D) a transient amplifier for $m \ge h+1$.}
        \label{fig:SI}
\end{figure}

 Our final example, the Star of Islands, is similar to Separated Hubs, except that the hubs are connected to each other. It consists of $h \geq 2$ hub vertices and $n \geq 2$ islands, with $m \geq 2$ vertices per island. The total population size is $N=nm+h$. Within the hub and each island, vertices are connected to one another with weight 1.  Each hub vertex is also connected to each island vertex with weight $\epsilon>0$.

For weak selection, in the $\epsilon \to 0$ limit, we calculate
\begin{equation}
\label{eq:SI_Ne}
\Ne= N + 
\frac{(m-h)mnh\big(h(h-1)+m(m-1)(n-2) \big)}
{\big(h(h-1)+m(m-1)\big)\big(h(h-1)+m(m-1)n\big)}.
\end{equation}
The second term on the right-hand side has the sign of $m-h$.  It follows that the Star of Islands is an amplifier of weak selection when $m>h$, and a suppressor of weak selection when $m<h$. 

We show in \ref{app:SIweak} that the strongest amplifier of weak selection occurs for $h=2$, $m=4$, and first $\epsilon \to 0$ and then $n \to \infty$. In this case $\Ne/N \to 9/7$. The strongest suppressor occurs for first $\epsilon \to 0$, then $n \to \infty$, and then $h \to \infty$, leading to $\Ne/N \to 0$.

For arbitrary $r>0$, in the  $\epsilon \to 0$ limit, we obtain $\rho_{SI_{n,m,h}(r)}  = \mathrm{num}/\mathrm{denom}$
with
\begin{multline}
\label{eq:SI_fixprob2}
    \mathrm{num}  = r^m(1-r^{-1})\left(1-r^{-(h+m)}\right)\\
    \left(hr^h\left(1-r^{-(h-1)}\right)\left(mn(m-1)r^m+h(h-1)\right)\right.\\
\left. + mr^m\left(1-r^{-(m-1)}\right)\left(mn(m-1)+h(h-1)r^h\right)\right),
\end{multline}
\begin{multline}
\mathrm{denom}  = (mn+h)\left(h\left(1-r^{-(h-1)}\right) \right.\\
\left.+mr^m\left(1-r^{-(m-1)}\right)\right)\left(mr^m\left(1-r^{-(m-1)}\right)(1-x^n)\right.\\
\left.+h\left(1-r^{-(h-1)}\right)\left(r^{h+m}-x^n\right)\right),
\end{multline}
and
\begin{equation}
\label{eq:SI_fixprob4}
x  = \frac{mr^{-m}(r^{m-1}-1)+h(r^{h-1}-1)}{mr^{h}(r^{m-1}-1)+h(r^{h-1}-1)}.
\end{equation}
In the limit of many islands, this simplifies to
\begin{equation}
\lim_{n \to \infty} \rho_{SI_{n,m,h}}(r) 
= \begin{cases}
0 & 0 \leq r \leq 1\\
\frac{(m-1)(1-r^{-1}) \left(1-r^{-(m-1)} \right)}{hr^{-m}\left(1-r^{-(h-1)}\right) + m\left(1-r^{-(m-1)}\right)}
& r>1.
\end{cases}
\end{equation}
We prove in \ref{app:SInonweak} that the Star of Islands is a reducer for $m=h$.

\subsection*{Approximating fixation probability}

We have defined the effective population size $\Ne$ in terms of the expected remeeting times of random walks.  While this definition allows $\Ne$---and, via Eq.~\eqref{eq:rhoexpand}, fixation probabilities under weak selection---to be computed in polynomial time, it gives little intuition for how $\Ne$ relates to more familiar graph statistics.  

To build such intuition, we use a mean-field approximation from Fotouhi et al.~\cite{fotouhi2019evolution}. We suppose that each remeeting time $\tau_i$ is approximately equal to a single value, $\tau$.  Then from Eq.~\eqref{eq:pi2tau} we have
\begin{align*}
1  = \sum_{i \in G} \pi_i^2 \tau_i
 \approx \tau \sum_{i \in G} \pi_i^2
 = \frac{\tau \sum_{i \in G} w_i^2}{(\sum_{i \in G} w_i)^2}
 = \frac{\tau \mu_2}{N \mu_1^2}.
\end{align*}
Above, $\mu_1= \frac{1}{N} \sum_{i \in G} w_i$ and $\mu_2= \frac{1}{N} \sum_{i \in G} w_i^2$ are the first and second moments, respectively, of the weighted degree distribution. Solving for $\tau$ and substituting in the definition of $\Ne$ gives the approximation
\begin{equation}
\label{eq:Neapprox}
\Ne \approx N \mu_1^2/\mu_2.
\end{equation}
Substituting in Eq.~\eqref{eq:rhoexpand} gives an approximation for fixation probability under weak selection in terms of $\mu_1$ and $\mu_2$.  Interestingly, the right-hand side of Eq.~\eqref{eq:Neapprox} was taken as the definition of effective population size by Antal et al.~\cite{antal2006evolutionary}, who studied the same model but arrived at this expression by different methods and assumptions.

\begin{figure}
    \centering
    \begin{tabular}{cc}
    \includegraphics[width=0.5\textwidth]{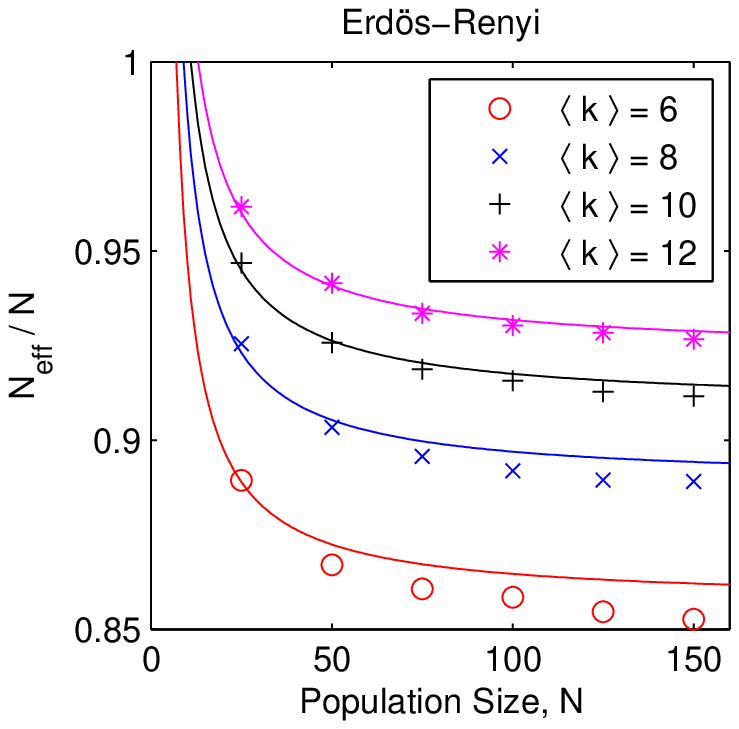} & \includegraphics[width=0.5\textwidth]{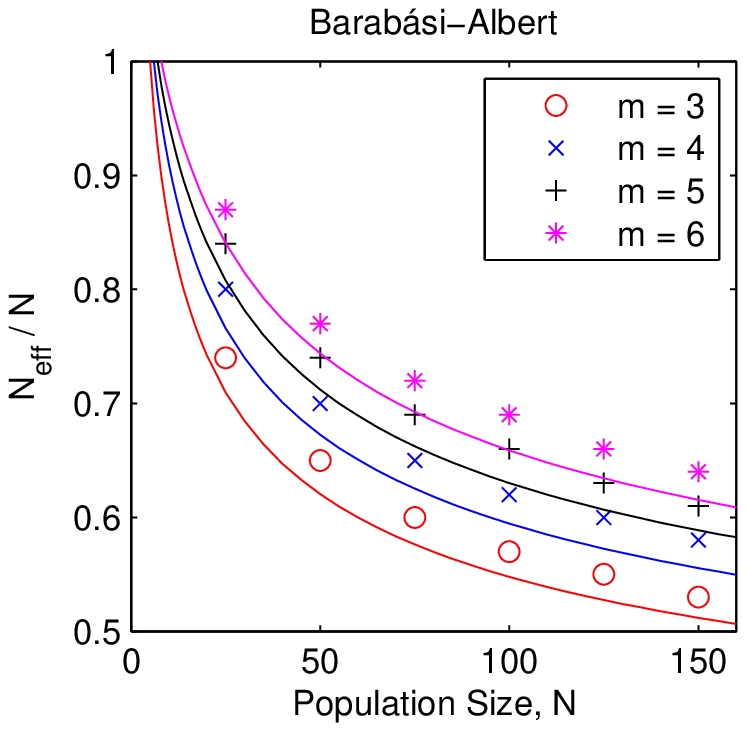}\\
    \textbf{\textsf{A}} & \textbf{\textsf{B}}
    \end{tabular}
    \caption{\textbf{Random graphs suppress weak selection.} Plot markers show the ratio $\Ne/N$, averaged over 1000 trials, plotted against population size $N$. Effective population size, $\Ne$, is calculated by numerically solving Eq.~\eqref{eq:trecur} for each graph and applying Eqs.~\eqref{eq:ti} and \eqref{eq:Ne}. All random graphs generated have $\Ne<N$ and are therefore suppressors of weak selection. Curves of the corresponding colors show the approximation $\Ne/N \approx \mu_1^2/\mu_2$ from Eq.~\eqref{eq:Neapprox}. Overall, we find that larger, sparser, and more heterogeneous graphs have smaller $\Ne/N$; these trends are all reflected in the approximation from Eq.~\eqref{eq:Neapprox}.
    \textbf{(A)} Erd\"{o}s-Renyi graphs were generated for specific values of the expected degree $\langle k \rangle$ by setting the link probability to $p=\langle k \rangle/(N-1)$. The moments $\mu_1$ and $\mu_2$ were approximated by assuming that the degree of each vertex is independently distributed as $\operatorname{Binom}(N-1,p)$.  This leads to $\Ne/N \approx (N-1)p/[(N-2)p+1]$. At the minimum population size of $N=\langle k \rangle + 1$, the graph is complete and therefore $\Ne/N=1$. \textbf{(B)} Barab\'{a}si-Albert preferential attachment networks \cite{barabasi1999emergence} were generated for linking numbers $3 \leq m \leq 6$, starting from a complete graph of size $m+2$.  The second moment was calculated using the expected degree distribution for finite Barab\'{a}si-Albert networks obtained by Fotouhi and Rabbat \cite{fotouhi2013network}. At the minimum population size of $N=m+2$, the graph is complete and therefore $\Ne/N=1$.}
    \label{fig:random}
\end{figure}

The approximation in Eq.~\eqref{eq:Neapprox} is reasonably accurate when compared to exact numerical calculation of $\Ne/N$ for Erd\"{o}s-Renyi and Barab\'{a}si-Albert graphs (Fig.~\ref{fig:random}).  In particular, the approximation explains the general trend that larger, sparser, and more heterogeneous graphs act as stronger suppressors (have smaller $\Ne/N$ ratio).  We note, however, that since $\mu_2 \leq \mu_1^2$ for any degree distribution,  the approximated $\Ne$ in Eq.~\eqref{eq:Neapprox} is at most equal to the actual population size $N$, with equality only for isothermal graphs. Therefore, the possibility of amplifiers of weak selection ($\Ne>N$) is not seen in this approximation.

\section*{Discussion}

\subsection*{Weak-selection methodology}

We have brought the method of weak selection, previously developed to analyze games on graphs \cite{Ohtsuki,Taylor,chen2013sharp,allen2014games,debarre2014social,durrett2014spatial,allen2017evolutionary,fotouhi2018conjoining}, to bear on the question of amplifiers and suppressors.  While our focus is on death-Birth updating, the method also applies to Birth-death updating, using a modified version of the coalescing random walk \cite{allen2017evolutionary,allen2018mathematical}. Our weak-selection method has the advantage of being computable in polynomial time (in the size of the graph), in contrast to other numerical methods \cite{hindersin2015most,hindersin2016exact,cuesta2017suppressors,cuesta2018evolutionary}, which take exponential time. Our expression for fixation probabilities in terms of coalescence times, Eq.~\eqref{eq:rhoexpand}, also enables the proof of general results such as the Weak-Selection Isothermal Theorem for dB.  A drawback of the weak-selection approach is that it does not distinguish between transient and non-transient amplifiers, nor can it detect complex behavior such as multiple switchings between amplification and suppression \cite{cuesta2018evolutionary}.

\subsection*{Effective population size}

Our analysis motivated a new definition of the effective population size of a graph, $\Ne=\sum_{i \in G} \pi_i \tau_i$. This notion of effective population is particular to dB updating, since it was derived from weak-selection fixation probabilities under this update rule. Our definition has a number of interesting connections to other definitions previously proposed for this concept \cite{antal2006evolutionary,broom2012two,allen2013adaptive,giaimo2018invasion}.  

First, as noted above, the effective population size of Antal et al.~\cite{antal2006evolutionary} appears in Eq.~\eqref{eq:Neapprox} as an approximation to ours. Whereas we obtain $\Ne \approx N \mu_1^2/\mu_2$ using coalescent theory and assuming uniformity of remeeting times, Antal et al.~\cite{antal2006evolutionary} obtain the same expression using diffusion approximation and assuming degree-uncorrelatedness of the graph. That the same expression arises from distinct analytical frameworks and assumptions hints at its naturality.

Second, our definition differs by a simple rescaling from the notion of ``fixation effective population size" proposed by Allen, Dieckmann, and Nowak (hereafter, ADN) \cite{allen2013adaptive}, and elaborated on by Giaimo et al.~\cite{giaimo2018invasion}:
\begin{equation}
\label{eq:AllenNe}
    \Ne^\mathrm{ADN} = \frac{N^2}{N-1} \frac{d \rho}{dr} \Big|_{r=1}.
\end{equation}
Comparing Eqs.~\eqref{eq:AllenNe} and \eqref{eq:rhoexpand}, we find the relationship
\begin{equation*}
    \Ne^\mathrm{ADN} = \frac{N(\Ne-2)}{2(N-1)}.
\end{equation*}
For large populations, $\Ne^\mathrm{ADN} \approx \Ne/2$.  The factor of two appears because the ADN definition uses the Wright-Fisher (discrete generations) model as a baseline, whereas the baseline for our $\Ne$ is the death-Birth process, for which generations are overlapping.  Such factors of two commonly appear in translating between discrete- and overlapping-generations models \cite{moran1962statistical,felsenstein1971inbreeding,allen2013adaptive}. 

Third, our proposed definition is closely related to the concept of ``inbreeding effective population size", which dates back to Wright \cite{wright1931evolution} and has been elaborated on by many others \cite{kimura1963measurement,felsenstein1971inbreeding,crow1988inbreeding,broom2012two}. The inbreeding effective population size is typically defined, for diploid populations, as the size of an idealized population that would have the same level of autozygosity (a locus containing two alleles that are identical by descent) \cite{kimura1963measurement,crow1988inbreeding}. Although autozygosity as such cannot occur in haploid populations, the remeeting time $\tau_i$ quantifies the closely-related concept of \emph{auto-coalescence}---the time for two hypothetical, independent lineages from $i$ to coalesce.  For rare mutation, coalescence time is proportional to the probability of non-identity by descent \cite{slatkin1991inbreeding}; thus auto-coalescence can be taken as a proxy for autozygosity in haploid populations. Our $\Ne$ is equal to the size of a well-mixed population that would experience the same degree of auto-coalescence, when averaged over individuals weighted by their reproductive values $\pi_i$.  It is therefore reasonable to interpret our $\Ne$ as a haploid analogue of the inbreeding effective population size.

\subsection*{Transient amplifiers of selection}

The most novel of our results is the discovery of the first transient amplifiers of selection for dB updating. Previous  investigations \cite{frean2008death,kaveh2015duality,hindersin2015most} had uncovered only suppressors and reducers.  Of the transient amplifiers we have found, the strongest is the 2-Fan, $F_{n,2}$, with many blades ($n \to \infty$; Fig.~\ref{fig:Fan}D). A companion work \cite{Pepa} proves that full (non-transient) amplifiers cannot exist for death-Birth updating.

Transient amplifiers appear to be quite rare for death-Birth updating. None were present within an ensemble of thousands of small graphs analyzed by Hindersin and Traulsen \cite{hindersin2015most}.  Similarly, no amplifiers of weak selection for dB were found in our ensembles of Erd\"{o}s-Renyi and Barabasi-Albert random graphs. 

Why should transient amplifiers be so rare? One possible clue comes from the  approximation for effective population size in Eq.~\eqref{eq:Neapprox}.  The approximated $\Ne$ is always less than or equal to the actual population size $N$, with equality only for isothermal graphs.  Thus any amplifier (transient or not) must be a graph for which the approximation in Eq.~\eqref{eq:Neapprox} is inaccurate.  Another possible clue is found by combining Eqs.~\eqref{eq:Ne} and \eqref{eq:pi2tau} to obtain
\begin{align*}
    \frac{N-\Ne}{N^2} = \left(\frac{1}{N} \sum_{i\in G}\pi_i^2\tau_i \right) - \left(\frac{1}{N} \sum_{i\in G}\pi_i \right) \left(\frac{1}{N} \sum_{i\in G}\pi_i\tau_i \right) .
\end{align*}
The right-hand side can be interpreted as the covariance of  $\pi_i$ with $\pi_i\tau_i$, as $i$ runs over vertices of $G$.  It follows that $G$ is an amplifier of weak selection if and only if $\pi_i$ and $\pi_i\tau_i$ are negatively correlated on $G$.  This requires a very strong negative relationship between weighted degree and remeeting time, which seems unlikely to arise in the usual random graph models.  A third clue comes from a companion work \cite{Pepa}, which proves a bound on the strength of transient amplifiers for dB.  Since transient amplifiers are limited in their possible strength, it is reasonable to suppose they are also limited in number. Each of these clues, however, falls very short of a formal proof.

\subsection*{Reducers of fixation}

Evolutionarily speaking, reducers of fixation maintain the status quo. They protect the resident type from replacement by any mutation, whether beneficial or deleterious.  Reducers may have applications in bio-engineering, in situations where it is desirable to inhibit the accumulation of all fitness-affecting mutations.  Indeed, it has been argued that the cycle-like structure of epithelial stem cells in mammals \cite{bozic2013unwanted,vermeulen2013defining} may have been evolutionarily designed to limit somatic mutations \cite{hindersin2016should}.  The cycle was the first known reducer \cite{hindersin2015most}; others were identified by Hindersin et al.~\cite{hindersin2016should}. To these examples we have added two more: the Separated Hubs graph with $n=h+1$ and the Star of Islands with $m=h$.

Isothermal graphs appear to be obvious candidates for reducers of fixation.  This is because, if $G$ is a reducer of fixation, then $\rho_G(r)$ and $\rho_{K_N}(r)$ must coincide to first order in $r$ at $r=1$, and this latter property holds for all isothermal graphs according to the Weak-Selection Isothermal Theorem for dB.  Indeed, all previously-known examples of reducers \cite{hindersin2015most,hindersin2016should} were isothermal. However, neither the Separated Hubs graph for $n=h+1$ nor the Star of Islands for $m=h$ are isothermal; thus reducers need not be isothermal. The converse question---whether all isothermal graphs are reducers---remains open.  To resolve this question, one would have to either discover or rule out other behaviors for isothermal graphs $G$, such as $\rho_G(r)>\rho_{K_N}(r)$ for all $r$ sufficiently close but not equal to 1.  Another open question is whether reducers of fixation exist for Bd updating.

\subsection*{Limitations}

Although we have uncovered an interesting range of behaviors for dB updating on graphs, there are limitations to our approaches.  All of our analytical results involve the limit of either weak selection or certain edge weights going to zero. Some of our results combine these limits, meaning that they apply only in rather extreme scenarios, and the results may depend on the limit ordering \cite{sample2017limits}.

We also do not consider the issue of fixation time \cite{frean2013effect,hindersin2014counterintuitive,askari2015analytical,hindersin2016exact,10.1371/journal.pcbi.1005864,hathcock2018fitness,moller2019exploring,tkadlec2019population}.  Previous work \cite{frean2013effect,hindersin2014counterintuitive,moller2019exploring,tkadlec2019population} has uncovered a tradeoff between fixation probability and time: Graphs that amplify selection tend to have larger fixation times than the complete graph, which impedes their ability to accelerate adaptation.  A number of our examples involve limits as certain edge weights go to zero.  Fixation times diverge to infinity for these examples; therefore they do not hasten the accumulation of beneficial mutations.  The search for graphs that (transiently) amplify selection without greatly increasing fixation times is left to future work.

\section*{Conclusion}

The identification of amplifiers and suppressors of selection has become a robust field of inquiry \cite{Lieberman2005Graphs,antal2006evolutionary,Broom2008Analysis,frean2008death,shakarian2012review,voorhees2013fixation,monk2014martingales,adlam2015amplifiers,cuesta2015fast,hindersin2015most,jamieson2015fixation,kaveh2015duality,pattni2015evolutionary,pavlogiannis2017amplification,cuesta2017suppressors,pavlogiannis2018construction,cuesta2018evolutionary,giaimo2018invasion,moller2019exploring}.  Most investigations of this question follow the lead of the initial work \cite{Lieberman2005Graphs} in focusing on Birth-death updating.  This is an interesting contrast to the study of games on graphs \cite{Ohtsuki,Taylor,chen2013sharp,debarre2014social,durrett2014spatial,pena2016evolutionary,allen2017evolutionary,fotouhi2018conjoining}, which typically considers death-Birth updating---likely because Birth-death updating tends not to support cooperative behaviors \cite{Ohtsuki}. 

Since the choice of update rule has such marked consequences, a full understanding of evolutionary dynamics in structured populations requires studying a variety of update rules.  Indeed, the update rule should properly be considered an aspect of the population structure, equal in importance to the graph itself \cite{debarre2014social,kaveh2015duality,pattni2015evolutionary,allen2018mathematical}.  If the theory of amplifiers and suppressors is to find application (for example, to microbial populations \cite{pavlogiannis2018construction}), it is critical to determine which update rules are plausible for specific organisms. Our work shows that dB updating exhibits at least some of the interesting phenomena that have been observed for Bd updating, and suggests there is more to be discovered.

\section*{Acknowledgments}
This work was supported by the National Science Foundation Award \#1715315.  We are grateful to Martin A.~Nowak, Josef Tkadlec, and Andreas Pavlogiannis for enlightening conversations.

\bibliographystyle{plos2015}
\bibliography{ConstantSel}

\appendix

\section{Model and notation}
\label{app:model}

\subsection{Graph}

Population structure is represented by a weighted graph $G$ with edge weights $w_{ij}, \; i,j \in G$. We require that $G$ be connected, undirected ($w_{ji} = w_{ij}$ for each $i,j \in G$), and have no self-loops ($w_{ii}=0$ for all $i$).

We define the \emph{weighted degree} of vertex $i$ as $w_i=\sum_{j \in G} w_{ij}$.  The total sum of all edge weights (counting both directions for each pair of vertices) is denoted $W$:
\[
W = \sum_{i,j \in G} w_{ij} = \sum_{i \in G} w_i.
\]
The \emph{relative weighted degree} of vertex $i$ is defined as $\pi_i = w_i/W$.

\subsection{Random walks}

Our analysis makes extensive use of random walks on graphs.  Steps are taken with probability proportional to edge weight. Thus the probability of a step from $i$ to $j$ is $p_{ij} = w_{ij}/w_i$.  The probability that an $n$-step walk from vertex $i$ terminates at vertex $j$ is denoted $p_{ij}^{(n)}$.  

Random walks on a weighted graph $G$ have a stationary probability distribution, in which the probability of being at vertex $i$ is equal to its relative weighted degree $\pi_i=w_i/\sum_j w_j$. This stationary distribution obeys the reversibility (detailed balance) property that for each $i,j \in G$, 
\begin{equation}
\label{eq:reverse}
\pi_i p_{ij}^{(n)} = \pi_j p_{ji}^{(n)}.
\end{equation}

\subsection{States and transitions}

There are two competing types: a resident (wild) type, labeled 0, and a mutant type, labeled 1. We indicate the type currently occupying vertex $i \in G$ by the variable $x_i \in \{0,1\}$.  The overall population state can be written as a binary vector $\vx = (x_i)_{i \in G} \in \{0,1\}^G$. 

Residents have fecundity (reproductive capacity) 1, while mutants have fecundity $r =1+\delta$.  Here, $\delta$ quantifies the reproductive advantage (or disadvantage, if negative) of mutants over residents.   The fecundity of vertex $i$ in a given state can be written compactly as $1 + \delta x_i$.  Neutral drift is represented by the case $r=1$, or equivalently, $\delta = 0$.

Under death-Birth updating, first an individual is selected, uniformly at random, to be replaced.  Then a neighbor is selected, with probability proportional to fecundity times edge weight, to reproduce into the vacancy.  Overall, the probability $e_{ij}(\vx)$ that vertex $j$ is replaced by the offspring of vertex $i$, in state $\vx$, can be written
\begin{equation}
\label{eq:dB}
e_{ij}(\vx) = \frac{1}{N}\frac{w_{ij} (1+\delta x_i)}{\sum_{k \in G} w_{kj} (1+\delta x_k)}.
\end{equation}
Offspring inherit the type of the parent.

\subsection{Fixation probability}

The process of resident-mutant competition is represented as a discrete-time Markov chain on $\{0,1\}^G$. This Markov chain has two absorbing states: the state $\mathbf{1}$ for which $x_i = 1$ for all $i \in G$, and the state $\mathbf{0}$ for which $x_i=0$ for all $i \in G$.  These states correspond to the fixation of mutants and residents, respectively.  All other states of the evolutionary Markov chain are transient \cite[Theorem 2]{allen2014measures}.  Thus from any given initial state, the evolutionary Markov chain will eventually become absorbed in either state $\mathbf{0}$ or state $\mathbf{1}$. 

As our initial condition, we suppose that a single mutant is placed on a vertex of $G$, chosen uniformly at random, and all vertices contain residents.  Unless otherwise stated, we will use this initial condition throughout.  We define the \emph{fixation probability} of mutants as the probability of absorption in state $\mathbf{1}$ from this initial condition.  The fixation probability of a mutant of fitness effect $r$ on graph $G$ is denoted $\rho_G(r)$; we will sometimes omit the $G$ for brevity.

\section{Method for weak selection}
\label{sec:weak}

Allen et al.~\cite{allen2017evolutionary} derived a weak-selection expansion for fixation probability in terms of coalescence times.  Here we provide an alternative derivation, based on results from Allen and McAvoy \cite{allen2018mathematical}.

\subsection{Change due to selection}

Each vertex has a \emph{reproductive value (RV)}, which quantifies its contribution, under neutral drift, to the future gene pool.  For death-Birth updating, the reproductive value of vertex $i$ is $N\pi_i$ \cite{maciejewski2014reproductive,Allen2015molecular,allen2017evolutionary,allen2018mathematical}.  The RV-weighted frequency of mutants in a given state $\vx$ is $\hat{x}=N\sum_{i\in G} \pi_i x_i$.  Our method for weak selection centers on the expected change, $\Delta(\vx)$, in $\hat{x}$ from a given state $\vx$.  Note that if vertex $j$ is replaced by the offspring of vertex $i$, the resulting change in $\hat{x}$ is $N\pi_j(x_i-x_j)$. Taking the expectation over all such events, with probabilities given by Eq.~\eqref{eq:dB}, we obtain 
\begin{align}
\nonumber
\Delta(\vx) & = N\sum_{i,j \in G} e_{ij}(\vx) \pi_j (x_i-x_j)\\
& = \sum_{j \in G} \pi_j\left( -x_j + \sum_{i \in G} x_i \frac{w_{ij} (1+\delta x_i)}{\sum_{k \in G} w_{kj} (1 + \delta x_k)}\right)\\
\label{eq:DeltadB}
& = \sum_{i \in G} x_i\left( -\pi_i + \sum_{j \in G} \pi_j \frac{w_{ij} (1+\delta x_i)}{\sum_{k \in G} w_{kj} (1 + \delta x_k)}\right).
\end{align}

To analyze weak selection, we form a first-order Taylor expansion of $\Delta(\vx)$ in $\delta$ around $\delta = 0$.  That is, we seek an expansion of the form
\begin{equation}
\label{eq:Deltaexpand}
\Delta(\vx) = \Delta^\circ(\vx) + \delta \Delta'(\vx) + \mathcal{O}(\delta^2).
\end{equation}
Above and throughout, we use a superscript $^\circ$ to indicate a quantity evaluated at neutral drift ($\delta=0$), and a prime $'$ to indicate a derivative with respect to $\delta$ at $\delta=0$.  Expanding Eq.~\eqref{eq:DeltadB} and making use of Eqs.~\eqref{eq:reverse}--\eqref{eq:hndef}, we obtain
\begin{align*}
\Delta(\vx)
& = \sum_{i \in G} x_i \left( - \pi_i + \sum_{j \in G}  \pi_j p_{ji}
\left( 1 + \delta \left( x_i - \sum_{k \in G} p_{jk} x_k \right) \right) \right)
+ \mathcal{O}(\delta^2)\\
& =  \sum_{i \in G} x_i \left( - \pi_i + \sum_{j \in G}  \pi_i p_{ij}
\left( 1 + \delta \left( x_i - \sum_{k \in G} p_{jk} x_k \right) \right) \right)
+ \mathcal{O}(\delta^2)\\
& = \delta \sum_{i \in G}  x_i \left( \pi_i x_i  - \pi_i\sum_{k \in G} p_{ik}^{(2)} x_k \right) + \mathcal{O}(\delta^2)\\
& =  \delta  \sum_{i \in G} \pi_i  x_i \left( x_i - x_i^{(2)} \right) + \mathcal{O}(\delta^2).
\end{align*}
Above, we have introduced the notation $x_i^{(2)}=\sum_{j \in G} p_{ij}^{(2)}x_j$. We conclude that $\Delta^\circ(\vx)=0$ for all states $\vx$, and that the first-order coefficient in Eq.~\eqref{eq:Deltaexpand} is given by 
\begin{equation}
\label{eq:D's}
\Delta'(\vx) = \sum_{i \in G} \pi_i x_i \left(x_i - x_i^{(2)} \right).
\end{equation}
The appearance of $x_i^{(2)}$ reflects the fact that, when a vertex is selected for replacement, the neighbors competing to fill the vacancy are two steps from each other.  Mutants co-occurring at distance two therefore affect each other's reproductive success.

\subsection{Fixation probability under weak selection}
\label{sec:weakfix}

We now turn to fixation probability.  The neutral ($\delta = 0$) fixation probability for dB updating is  $\rho^\circ = \frac{1}{N}$ \cite{Allen2015molecular}. We therefore seek a weak-selection expansion of the form
\begin{equation}
\label{eq:rhoexpandapp}
\rho(1+\delta) = \frac{1}{N} + \delta \rho' + \mathcal{O}(\delta^2).
\end{equation}
To obtain an expression for $\rho'$, we introduce a small rate of mutation $u>0$ (which we will later take to zero). With each reproduction, the offspring inherits the type of the parent with probability $1-u$; otherwise, with probability $u$, the offspring is assigned either type 0 or 1 with equal probability. With mutation, the evolutionary Markov chain becomes ergodic \cite[Theorem 1]{allen2014measures}, with a unique stationary probability distribution which we call the \emph{Mutation-Selection Stationary (MSS) distribution}.  We denote expectations in this distribution by $\E_\MSS$.

Allen and McAvoy \cite{allen2018mathematical} obtained a relationship between the expectation of $\Delta(\vx)$ under the MSS distribution, and the fixation probabilities of two competing types (with each invading the other). In the case of mutant-resident competition in the death-Birth process, Eqs.~(57)-(59) of Allen and McAvoy \cite{allen2018mathematical} yield 
\begin{equation}
\label{eq:MSSrho}
\lim_{u \to 0} \frac{\E_\MSS[\Delta]}{u}
= \frac{\rho(r)-\rho(r^{-1})}{2 \left(\rho(r) + \rho(r^{-1}) \right)}.
\end{equation}
We note that $\rho(r^{-1})$ is the fixation probability of the resident type into a mutant-dominated population, with the initial resident placed uniformly at random.  This is because a resident of fitness 1 invading a mutant population of fitness $r$ is equivalent (upon rescaling by $r$) to a mutant of fitness $r^{-1}$ invading a resident population of fitness 1.

For neutral drift ($r=1$), we have $\Delta^\circ(\vx)=0$ for all states $\vx$, and $\rho^\circ=\frac{1}{N}$ \cite{Allen2015molecular}. Taking the derivative of both sides of Eq.~\eqref{eq:MSSrho} at $r=1$ yields
\begin{equation}
\label{eq:MSSrhoweak}
\lim_{u \to 0} \frac{\E^\circ_\MSS[\Delta']}{u}
= \frac{\rho'}{2 \rho^\circ} = \frac{N\rho'}{2}.
\end{equation}
The notation $\E_\MSS^\circ$ means that the expectation is taken over the \emph{neutral} ($\delta=0$) MSS distribution. Combining Eqs.~\eqref{eq:D's} and \eqref{eq:MSSrhoweak}, we have
\begin{equation}
\label{eq:rhoxixj}
\rho' = \frac{2}{N} \lim_{u \to 0} \frac{\E^\circ_\MSS[\Delta']}{u}
=  \frac{2}{N} \sum_{i \in G} \pi_i \lim_{u \to 0} \left( \frac{\E^\circ_\MSS \left [x_i \left(x_i - x_i^{(2)}\right) \right]}{u}  \right).
\end{equation}

\subsection{Reduction to coalescence times}

Eq.~\eqref{eq:rhoxixj} expresses $\rho'$ in terms of a particular statistic of spatial assortment under the MSS distribution.  Such statistics can be calculated in terms of the coalescence times $\tau_{ij}$, which are the unique solution to
\begin{equation}
\label{eq:trecurapp}
\tau_{ij} = \begin{cases} 0 & i=j\\
1 + \frac{1}{2}  \sum_{k \in G} \left(p_{ik} \tau_{jk} + p_{jk} \tau_{ik} \right)& i \neq j.
\end{cases}
\end{equation}
Intuitively, the larger the coalescence time $\tau_{ij}$, the more that the occupants of $i$ and $j$ are separated from their common ancestor, and the less likely the are to have the same type. Specifically, Eq.~(111) of Allen and McAvoy \cite{allen2018mathematical} gives the relationship
\begin{equation}
\label{eq:MSStau}
\lim_{u \to 0} \frac{\E_\MSS^\circ[x_i(x_i - x_j)] }{u} = \frac{\tau_{ij}}{4}.
\end{equation}
Combining Eqs.~\eqref{eq:MSSrhoweak}, \eqref{eq:MSStau}, and \eqref{eq:tndef}, we obtain
\begin{equation}
\label{eq:rho't2}
\rho' = \frac{\tau^{(2)}}{2N},
\end{equation}
where $\tau^{(n)}$  is the expected coalescence time from the two ends of an $n$-step stationary random walk:
\begin{equation}
\label{eq:tndef}
\tau^{(n)} = \sum_{i.j \in G} \pi_ip_{ij}^{(n)}\tau_{ij}
\end{equation}
 We observe that $\tau^{(2)}$ characterizes the frequency with which mutants co-occur (and therefore compete with each other) at distance two.

Allen et al.~\cite{allen2017evolutionary} derived the following recurrence relation for $\tau^{(n)}$:
\begin{equation}
\label{eq:tnrecur}
\tau^{(n+1)} = \tau^{(n)} + \sum_{i \in G}\pi_i p_{ii}^{(n)} \tau_{i} - 1.
\end{equation}
Above, $\tau_i$ is the remeeting time from vertex $i$:
\begin{equation}
\label{eq:tidef}
\tau_i = 1 + \sum_{j \in G} p_{ij} \tau_{ij}.
\end{equation}
 Noting that $p_{ii}^{(0)}=1$ and $p_{ii}^{(1)}=0$ (since $G$ has no self-loops), we have
\begin{align}
\label{eq:t0}
\tau^{(0)} & = 0\\
\label{eq:t1}
\tau^{(1)} & = \sum_{i \in G}\pi_i \tau_{i} - 1\\
\label{eq:t2}
\tau^{(2)} & = \sum_{i \in G}\pi_i \tau_{i} - 2.
\end{align}
As in the main text, we define effective population size as 
\begin{equation}
\label{eq:Nedef}
\Ne = \sum_{i \in G} \pi_i \tau_i.
\end{equation} 
Combining Eqs.~\eqref{eq:rhoexpandapp}, \eqref{eq:rho't2}, \eqref{eq:t2}, and \eqref{eq:Nedef} we obtain the expansion
\begin{equation}
\label{eq:rhoNe}
\rho(1+\delta) = \frac{1}{N} + \delta \frac{\Ne-2}{2N} + \mathcal{O}(\delta^2),
\end{equation}
which is Eq.~(8) of the main text.

Thus a weak-selection expansion of fixation probabilities can be obtained for any graph by first computing pairwise coalescence times, and using these to compute $\Ne$. This method can be performed in $\mathcal{O}(N^6)$ time using standard methods such as Gaussian elimination.

\section{Mathematical lemmas}
\label{app:lemmas}

Here we prove mathematical lemmas that will be used to characterize the behavior of our example families in graphs in Appendix \ref{sec:examples}.  In analyzing these examples, we repeatedly encounter the function
\begin{equation}
F_r(x) = \ln{\left|\frac{1-r^{-x}}{x}\right|},
\end{equation}
defined for $x>0, r>0, r \neq 1$. Before proceeding, we establish some results about this function and about convex functions in general.

\begin{lemma}\label{lemma_F}
$F_{r}(x)$ is increasing in $x$ for $0<r<1$, decreasing in $x$ for $r>1$, and convex in $x$ for all $r>0, r \neq 1$.
\end{lemma}

\begin{proof}
To prove the monotonicity of $F_{r}(x)$, we examine the first derivative:
\begin{equation*}
 F_{r}'(x)=\frac{\ln{\left(r^x\right)}}{x\left(r^{x}-1\right)}-\frac{1}{x}.
\end{equation*}

We note that for all $y>0, y\ne 1$,
\begin{equation}
\label{eq:logineq}
\ln{(y)}<y-1.
\end{equation}
In particular, 
\begin{equation}
\label{eq:rxineq}
\ln{(r^{x})}<r^{x}-1.
\end{equation}
In the case $r>1$ (recalling that $x>0$) we have $x(r^{x}-1)>0$.  Dividing inequality~\eqref{eq:rxineq} by $x(r^{x}-1)$ we obtain
\begin{align*}
\frac{\ln{\left(r^x\right)}}{x\left(r^{x}-1\right)}&>\frac{1}{x},
\end{align*}
which shows that $F'_r(x) > 0$.  In the case $0<r<1$ we have $x(r^{x}-1)<0$ and therefore 
\begin{align*}
\frac{\ln{\left(r^x\right)}}{x\left(r^{x}-1\right)}&<\frac{1}{x},
\end{align*}
which shows that $F'_r(x) < 0$.  This proves that $F_r(x)$ is increasing in $x$ for $r>1$ and decreasing for $0 < r < 1$.

To prove the convexity of $F_{r}(x)$, we examine the second derivative,
\begin{equation}
\label{eq:F''}
 F_{r}''(x)=\frac{\left(r^{x}-1\right)^2-r^{x}\ln^{2}{\left(r^x\right)}}{x^2(r^x-1)^2}.
\end{equation}
The denominator is positive for all valid $r$ and $x$, so we must show the numerator is positive as well. We rearrange inequality Eq.~\eqref{eq:logineq} to
\[
2y>2(1+\ln y).
\]
For $r>1$, integrating both sides from $y=1$ to $y=r^{x/2}$ gives 
\[
r^x-1 > r^{x/2} \ln{(r^x)}
\]
Squaring both sides gives 
\begin{equation}
(r^{x}-1)^{2}>r^{x}\ln^{2}{(r^{x})}.
\end{equation}
A similar argument obtains the same result for the case $0<r<1$. Thus the numerator of $F_{r}''(x)$ in Eq.~\eqref{eq:F''} is positive. This proves that $F_r(x)$ is strictly convex in $x$ for all $r>0, r\neq 1$.
\end{proof}

\begin{lemma}
\label{lem:convex}
Let $f(x)$ be a twice-differentiable, strictly convex function defined on some interval $[a,b]$.  Then for any $d$ such that $0<d<(b-a)/2$,
\begin{equation}
f(a+d)+f(b-d) < f(a) + f(b).
\end{equation}
\end{lemma}

\begin{proof}
By the Mean Value Theorem, there exists $c_1 \in (a,a+d)$ such that
\begin{equation}
\label{eq:MVT1}
f'(c_{1})= \frac{f(a+d)-f(a)}{d}.
\end{equation}
Similarly, there exists $c_2 \in (b-d,b)$ such that
\begin{equation}
\label{eq:MVT2}
f'(c_{2})= \frac{f(b)-f(b-d)}{d}.
\end{equation}
Because $d<(b-a)/2$, $c_{2}>c_{1}$. Since $f$ is strictly convex, its derivative is strictly increasing. Thus, 
$f'(c_{2})>f'(c_{1})$.  Combining with Eqs. \eqref{eq:MVT1} and \eqref{eq:MVT2} yields
\[
f(b)-f(b-d)>f(a+d)-f(a).
\]
Rearranging yields the desired result.
\end{proof}

\section{Examples}
\label{sec:examples}

We consider three examples: the $m$-Fan, Star of Islands, and Star of Islands with Separated Hubs.  These graphs act conditionally as reducers, suppressors, or transient amplifiers. 

\subsection{Separated Hubs}

The Separated Hubs graph has $h\geq 1$ ``hub'' vertices, $n\geq 2$ ``blades'', and $m \geq 2$ vertices per blade. Within each blade, vertices are connected to one another with weight 1. Each hub vertex is also connected to each blade vertex with weight $\epsilon$. The Fan graph is the $h=1$ case of Separated Hubs.

\subsubsection{Weak Selection}

Let $H$ represent a hub vertex and $B$ represent any blade vertex. The weighted degree and relative weighted degrees of these vertices are:
\begin{align*}
      w_H&=nm\epsilon\\
w_B&=m-1+h\epsilon\\
\pi_H&=\frac{w_{H}}{h w_{h}+n m w_{B}}=\frac{\epsilon }{2 h \epsilon +m-1}\\
\pi_B&=\frac{w_{B}}{h w_{h}+n m w_{B}}=\frac{h \epsilon +m-1}{n m(2 h \epsilon +m-1)}.
  \end{align*}
Let $B$ and $B'$ represent any two vertices on the same blade and $H$ and $H'$ represent any two hub vertices. The step probabilities are:
\begin{align}
    p_{BB'}&=\frac{1}{w_{B}}=\frac{1}{m-1+h\epsilon}\\
    p_{BH}&=\frac{\epsilon}{w_{B}}=\frac{\epsilon}{m-1+h\epsilon}\\
    p_{HB}&=\frac{\epsilon}{w_{H}}=\frac{1}{nm}.
\end{align}
All other step probabilities (e.g.~between vertices on different blades) are zero.

There are four coalescence times to consider: $\tau_{HH'}$ for two different hub vertices, $\tau_{HB}$ for a hub and a blade, $\tau_{BB'}$ for two different vertices on a common blade, and $\tau_{BB''}$ for two vertices on different blades.  Eq.~\eqref{eq:trecurapp} becomes
\begin{align}
    \tau_{HH'}&=1+\frac{1}{2}\big(2nm p_{HB} \big)\\
    \tau_{BH}&=1+\frac{1}{2}\big((h-1) p_{BH}\tau_{HH'}+(m-1) p_{BH}+(m-1) p_{HB} \tau_{BB'}+(n-1)m p_{HB} \tau_{BB'} \big)\\
    \tau_{BB'}&=1+\frac{1}{2}\big( 2(m-2) p_{BB'} \tau_{BB'}+2h p_{BH} \tau_{BH}\big)\\
    \tau_{BB''}&=1+\frac{1}{2}\big(2(m-1) p_{BB''}+2h p_{BH} \tau_{BH}\big).
\end{align}

Solving this system and substituting into Eqs.~\eqref{eq:tidef} and \eqref{eq:Nedef}, we obtain $\Ne=\text{num/denom}$, where 
\begin{multline}
\label{eq:SHNenum}
    \mathrm{num}=m^{4}n+h\epsilon\left(h\epsilon-1\right)+m^{3}\left(n(6h\epsilon+\epsilon-1)-1\right)\\
    +m^{2}\left(2-2h\epsilon+n\left(-1-(1+4h)\epsilon+h(1+12h)\epsilon^2\right)\right)\\
    +m\left(-1+3h\epsilon-h^2\epsilon^2+n\left(1-4h^2\epsilon^2+8h^3\epsilon^3-h\epsilon(2+\epsilon))\right)\right)
\end{multline}
\begin{equation}
\label{eq:SHNedenom}
 \mathrm{denom}=\left((-1+m+2h\epsilon)\left(m^2+h\epsilon(-1+h\epsilon)+m(-1+n\epsilon+h\epsilon(2+n\epsilon))\right)\right).
\end{equation}
For $\epsilon \to 0$, the effective population size becomes
\begin{equation}
\label{eq:SHNe}
        \lim_{\epsilon \to 0}\Ne = nm + n - 1.
\end{equation}

Since the population size is $N=nm+h$, we find that the Separated Hubs graph is an amplifier of weak selection for $n>h+1$ and a suppressor of weak selection for $n<h+1$.  We show below that these results extend to nonweak selection as well.

From Eq.~\eqref{eq:rhoNe}, we have the first-order term of the weak-selection fixation probability:
\begin{equation}
    \lim_{\epsilon\rightarrow 0}\rho'=\frac{n m+n-3}{2 (h+ nm)}.
\end{equation}
Since the right-hand side is decreasing in $h$, the best amplifier of weak selection occurs when $h=1$, which is the case of the Fan (section \ref{sect:fan}).

\subsubsection{Nonweak Selection}
\label{sec:SHnonweak}

We now calculate fixation probabilities, for arbitrary mutant fitness $r$, on the Separated Hubs graph in the $\epsilon \to 0$ limit. This limit leads to a separation of timescales: a fast timescale for events whose probability is $\mathcal{O}(1)$ as $\epsilon \to 0$, and a slow timescale for events with probability $\mathcal{O}(\epsilon)$. Fixation on each blade vertex occurs on the fast timescale.  Replacement of hub vertices by individuals on the blades also occurs on the fast timescale.  Changes in the number of blade vertices that are fixed for mutants occurs on the slow timescale.

Since replacement of hub vertices occurs on the fast timescale, we assume that the types of the hubs converge in time-average to their stationary probability distribution.  To determine this distribution, we note that, when a hub is chosen for replacement, the probability to be replaced by a mutant is proportional to $rkm$, while the probability to be replaced by a resident is proportional to $n-k$, where $k$ denotes the number of blades that are fixed for mutants. Therefore, for each hub vertex, the stationary probabilities are given by
\begin{align*}
    \Prob[\text{Hub is M}] & = \frac{rk}{rk + n-k}\\
    \Prob[\text{Hub is R}] & = \frac{n-k}{rk + n-k}.
\end{align*}
The types of different hub vertices are independent in the stationary distribution.

Dynamics on the slow timescale can be represented as a continuous-time Markov chain. We identify the states according to the number $k \in \{0, \ldots, n\}$ of blades that are fixed for mutants (with the remaining $n-k$ fixed for residents). Let us derive the transition rate from state $k$ to state $k+1$.  First, a resident blade vertex must be chosen for replacement, which happens with probability 
\begin{equation}
\label{eq:probBRdies}
\Prob[\text{Blade R dies}] = \frac{m(n-k)}{nm+h}.
\end{equation}
Then a mutant from the hub must be chosen to replace this resident blade vertex. The probability of this is
\[
\Prob[\text{Hub M replaces}] = \E \left[ \frac{\epsilon r X}{\epsilon(rX + h-X) + m-1} \right] .
\]
Above, $X$ is a random variable representing the number of mutant hubs in this state, with distribution
\[
X \sim \operatorname{Binom} \left(h, \frac{rk}{rk + n-k} \right).
\]
As $\epsilon \to 0$, the probability of a hub mutant replacing the vacancy becomes 
\begin{equation}
\label{eq:probHMreplaces}
\Prob[\text{Hub M replaces}] =  \frac{\epsilon r^2 h k}
{(rk+n-k)(m-1)} + \mathcal{O}(\epsilon^2).
\end{equation}
Finally, the mutant type must become fixed on the blade.  Since each blade is a complete graph of size $m$, this happens with probability
\begin{equation}
\label{eq:probBMfixes}
\Prob[\text{M fixes on blade}] =  \frac{m-1}
{m} \frac{1-r^{-1}}{1-r^{-(m-1)}}.
\end{equation}
Combining Eqs.~\eqref{eq:probBRdies}, \eqref{eq:probHMreplaces}, and \eqref{eq:probBMfixes} and neglecting terms of order $\epsilon^2$, the transition rate from state $k$ to $k+1$ for the slow timescale is
\begin{equation}
Q_{k, k+1} = \left(\frac{m(n-k)}{nm+h} \right) 
\left(\frac{\epsilon r^2 h k}{(rk+n-k)(m-1)}\right)
\left(\frac{m-1}
{m} \frac{1-r^{-1}}{1-r^{-(m-1)}} \right).
\end{equation}
A similar argument yields the transition rate from state $k$ to $k-1$:
\begin{equation}
Q_{k, k-1}=\left(\frac{mk}{nm+h} \right) 
\left(\frac{\epsilon h (n-k)}{(rk+n-k)r(m-1)}\right)
\left(\frac{m-1}
{m} \frac{1-r}{1-r^{m-1}} \right).
\end{equation}
We observe that
\begin{equation}
    \frac{Q_{k, k-1}}{Q_{k, k+1}}=r^{-(m+1)},
\end{equation}
independently of $k$.

The probability of fixation starting with one blade is:
\begin{align*}
\lim_{t \to \infty} Q^{(t)}_{1, n}
    &=\frac{1}{1+\sum_{k=1}^{n-1}\prod_{j=1}^{k}\frac{Q_{k, k-1}}{Q_{k, k+1}}}\\
    &=\frac{1}{1+\sum_{k=1}^{n-1} r^{-k(m+1)}}\\
    & =\frac{1-r^{-(m+1)}}{1-r^{-n(m+1)}}.
\end{align*}
In the $\epsilon \to 0$ limit, fixation becomes impossible for mutants originating on a hub vertex, since these vertices are much more likely to be replaced than to be chosen for reproduction.  Thus the overall fixation probability is equal to the probability that a mutant is placed on a blade, multiplied by the probability that it fixates on the blade, multiplied by the probability of fixation from one blade:
 \begin{align}
 \nonumber
  \rho_{SH_{n,m,h}}(r) & =\left(\frac{nm}{nm+h}\right) \left(\frac{m-1}{m}\frac{1-r^{-1}}{1-r^{-(m-1)}}\right)
  \left(\frac{1-r^{-(m+1)}}{1-r^{-n(m+1)}}\right)\\
  \label{eq:SH_fixprobapp}
  & =\left(\frac{n(m-1)}{nm+h}\right) \left(\frac{1-r^{-1}}{1-r^{-(m-1)}}\right)
  \left(\frac{1-r^{-(m+1)}}{1-r^{-n(m+1)}}\right).
\end{align}

In the limit of many blades, we obtain
\begin{equation}
\lim_{n \to \infty} \rho_{SH_{n,m,h}}(r) =
\begin{cases}
0 & 0\leq r\leq 1\\
\frac{(m-1)(1-r^{-1})(1-r^{-(m+1)})}{m(1-r^{-(m-1)})}
 &r>1.
\end{cases}
\end{equation}
Interestingly, this limit is independent of the number $h$ of hubs.

We are now prepared to prove a complete classification of the behavior of the Separated Hubs graph in the $\epsilon \to 0$ limit:

\begin{theorem}
\label{thm:SHclassify}
The Separated Hubs graph $SH_{m,n,h}$ is, in the $\epsilon \to 0$ limit,
\begin{itemize}
\item a suppressor for $n \le h$, 
\item a reducer for $n=h+1$, 
\item a transient amplifier for $n \ge h+2$.
\end{itemize}
\end{theorem}

\begin{proof}
We wish to compare the fixation probability on the Separated Hubs graph to that of a complete graph of equal size $N=nm+h$.  The ratio of fixation probabilities is 
\begin{equation}
\label{eq:SH_fixprob_ratio}
 \frac{\rho_{SH_{n,m,h}}(r)}{\rho_{K_{nm+h}}(r)}=
 \frac{\displaystyle \left(\frac{1-r^{-(m+1)}}{m+1}\right)\left(\frac{1-r^{-(nm+h-1)}}{nm+h-1}\right)}
 {\displaystyle \left(\frac{1-r^{-(m-1)}}{m-1}\right)\left(\frac{1-r^{-n(m+1)}}{n(m+1)}\right)}.
\end{equation}
The logarithm of this ratio can be written in terms of the function $F_{r}(x)=\ln\left|\frac{1-r^{-x}}{x}\right|$:
\begin{equation}
\label{eq:SH_fixprob_ratio_F}
\ln\left(\frac{\rho_{SH_{n,m,h}}(r)}{\rho_{K_{nm+h}}(r)}\right)= F_{r}(m+1)+F_{r}(nm+h-1)-F_{r}(m-1)-F_{r}(nm +n).
\end{equation}
We prove the cases in increasing order of difficulty.

\begin{case}{Case 1, $n=h+1$}
Substituting $h=n-1$ into Eq.~\eqref{eq:SH_fixprob_ratio_F} gives
\[
\ln\left(\frac{\rho_{SH_{n,m,n-1}}(r)}{\rho_{K_{nm+n-1}}(r)}\right) = F_{r}(m+1)+F_{r}(nm+n-2)-F_{r}(m-1)-F_{r}(nm + n).
\]
By Lemma \ref{lemma_F} $F_r(x)$ is convex in $x$ for all $r \ne 1$.  Applying Lemma \ref{lem:convex} with $a=m-1$, $b=nm+m$, and $d=2$, we obtain
\begin{equation}
\label{eq:SH_reducer_Fineq}
\ln\left(\frac{\rho_{SH_{n,m,n-1}}(r)}{\rho_{K_{nm+n-1}}(r)}\right) =  F_r(m+1)+F_r(nm+n-2) - F_r(m-1)-F_r(nm+n) < 0.
\end{equation}
Therefore $\rho_{SH_{n,m,n-1}}(r) < \rho_{K_{nm+n-1}}(r)$ for all $r \ne 1$, proving $SH_{n,m,h}$ is a reducer in this case. 
\end{case}

\begin{case}{Case 2, $n \le h$}  Here we wish to prove that $SH_{n,m,h}$ is a suppressor. We break into three subcases.

\begin{case}{Subcase 2.1, $r>1$}  By Lemma \ref{lemma_F}, $F_r(x)$ is decreasing in $x$ for $r>1$. Therefore, for $n \le h$ and $r>1$,
\begin{equation}
\label{eq:SH_suppressor_dec_rbig}
F_{r}(nm+h-1) < F_{r}(nm+n-2).
\end{equation}
Combining with the inequality in Eq.~\eqref{eq:SH_reducer_Fineq}, we have
\begin{align*}
\ln\left(\frac{\rho_{SH_{n,m,h}}(r)}{\rho_{K_{nm+h}}(r)}\right) = F_{r}(m+1)+F_{r}(nm+h-1)-F_{r}(m-1)-F_{r}(nm+n)<0.
\end{align*}
Thus $\rho_{SH_{n,m,h}}(r) < \rho_{K_{nm+h}}(r)$ for $n \le h$ and $r>1$. 
\end{case}

\begin{case}{Subcase 2.2, $0<r<1$ and $n \le h-1$} By Lemma \ref{lemma_F}, $F_r(x)$ is increasing in $x$ for $0<r<1$. Therefore, in this subcase, we have $F_{r}(m+1) > F_{r}(m-1)$ and $F_{r}(nm+h-1) \ge F_{r}(nm+n)$. Combining these two inequalities gives
\begin{equation}
\ln\left(\frac{\rho_{SH_{n,m,h}}(r)}{\rho_{K_{nm+h}}(r)}\right)
= F_{r}(m+1)+F_{r}(nm+h-1)-F_{r}(m-1)-F_{r}(nm+n)>0.
\end{equation}
Thus $\rho_{SH_{n,m,h}}(r) > \rho_{K_{nm+h}}(r)$ for $0<r<1$ and $n \le h-1$.
\end{case}

\begin{case}{Subcase 2.3, $0<r<1$ and $n=h$}  Substituting $h=n$ into Eq.~\eqref{eq:SH_fixprob_ratio_F} gives  
\begin{equation}
\label{eq:SH_suppressor_ineq}
\ln\left(\frac{\rho_{SH_{n,m,n}}(r)}{\rho_{K_{nm+n}}(r)}\right) =  F_{r}(m+1)-F_{r}(m-1)-\left(F_{r}(nm+n)-F_{r}(nm+n-1)\right).
\end{equation}
 
First consider the difference $F_{r}(m+1)-F_{r}(m-1)$. By Lemma \ref{lemma_F}, $F_r(x)$ is convex and increasing when $0<r<1$. Thus, $F_{r}(m+1)-F_{r}(m-1)$ is minimized when $m$ is smallest ($m=2$): 
\begin{align}
  F_{r}(m+1)-F_{r}(m-1)&\ge F_{r}(3)-F_{r}(1).
  \label{eq:SH_suppressor_ineqLHS}
  \end{align}
    
Next consider the difference $F_{r}(nm+n)-F_{r}(nm+n-1)$.  By the Mean Value Theorem, there is some $c\in(nm+n-1, nm+n)$ such that
\begin{equation}
\label{eq:MVTsuppressor}
F_r(nm+n) - F_r(nm+n-1) = F_{r}'(c).
\end{equation}
Since $F_{r}'(x)$ is increasing in $x$ for $0<r<1$ by Lemma \ref{lemma_F}, we have
\begin{equation*}
F_{r}'(c) < \lim_{x\to\infty}F'_r(x)=\lim_{x\to\infty}\left(\frac{\ln{\left(r^x\right)}}{x\left(r^{x}-1\right)}-\frac{1}{x}\right)=-\ln{(r)}.
\end{equation*}
Combining with Eq.~\eqref{eq:MVTsuppressor} gives
\begin{equation}
  F_{r}(nm+n)-F_{r}(nm+n-1)< -\ln{(r)}.\label{eq:SH_suppressor_ineqRHS}
  \end{equation}
 Combining Eqs.~\eqref{eq:SH_suppressor_ineq}, \eqref{eq:SH_suppressor_ineqLHS}, and \eqref{eq:SH_suppressor_ineqRHS} gives the inequality
      \begin{equation*}
\ln\left(\frac{\rho_{SH_{n,m,h}}(r)}{\rho_{K_{nm+h}}(r)}\right) > F_{r}(3)-F_{r}(1)+\ln{r}.
      \end{equation*}
      The right-hand side above is equal to $\ln{\left(1+\frac{\left(r-1\right)^2}{3r}\right)}$, which is positive for $0<r<1$.  Therefore $\rho_{SH_{n,m,h}}(r) > \rho_{K_{nm+h}}(r)$ for $0<r<1$ and $n = h$.
\end{case}
      
Combining the three subcases of Case 2, we have shown that when $n \le h$, we have $\rho_{SH_{n,m,h}}(r) > \rho_{K_{nm+h}}(r)$ for $0<r<1$ and $\rho_{SH_{n,m,h}}(r) < \rho_{K_{nm+h}}(r)$ for $r>1$.  This proves that the Separated Hubs graph is a suppressor for $n \le h$.
\end{case}

\begin{case}{Case 3, $n \ge h+2$} Here we wish to prove that $SH_{n,m,h}$ is a transient amplifier. We break into three subcases.

\begin{case}{Subcase 3.1, $0 < r < 1$} By Lemma \ref{lemma_F}, $F_r(x)$ is increasing in $x$ for $0<r<1$.  Since $h \le n-2$, we have $F_{r}(nm+h-1) < F_{r}(nm+n-2)$.  Combining with Eqs.~\eqref{eq:SH_fixprob_ratio_F} and \eqref{eq:SH_reducer_Fineq} gives
\begin{equation}
\ln\left(\frac{\rho_{SH_{n,m,h}}(r)}{\rho_{K_{nm+h}}(r)}\right)= F_r(m+1)+F_r(nm+h-1) - F_r(m-1)-F_r(nm+n) < 0.
\end{equation}
This proves that $\rho_{SH_{n,m,h}}(r)<\rho_{K_{nm+h}}(r)$ for $n \ge h+2$ and $0<r<1$.
\end{case}

\begin{case}{Subcase 3.2, $r>1$} To complete the proof that the Separated Hubs graph is a transient amplifier for $n \ge h+2$, we must show that $\rho_{SH_{n,m,h}}(r) > \rho_{K_{nm+n-1}}(r)$ only on some connected interval $1<r<r^*$.  We already know that $\rho_{SH_{n,m,h}}(1) = \rho_{K_{nm+h}}(1)=1/N$, and that $\rho_{SH_{n,m,h}}$ is an amplifier of weak selection, meaning that $\rho_{SH_{n,m,h}}(r) > \rho_{K_{nm+h}}(r)$ for some interval $1<r<r^*$.  It remains only to prove that the set of $r$-values for which $\rho_{SH_{n,m,h}}(r) > \rho_{K_{nm+h}}(r)$ is connected.  This amounts to showing that the equation $\rho_{SH_{n,m,h}}(r) = \rho_{K_{nm+h}}(r)$ has only one solution with $r>1$.

To show this, we set the ratio of fixation probabilities equal to one:
\begin{equation}
1 =  \frac{\rho_{SH_{n,m,h}}(r)}{\rho_{K_{nm+h}}(r)}=
 \frac{\displaystyle \left(\frac{1-r^{-(m+1)}}{m+1}\right)\left(\frac{1-r^{-(nm+h-1)}}{nm+h-1}\right)}
 {\displaystyle \left(\frac{1-r^{-(m-1)}}{m-1}\right)\left(\frac{1-r^{-n(m+1)}}{n(m+1)}\right)}.
\end{equation}
This is equivalent to 
\begin{align}\label{eq:SH_transamp}
(nm+h-1)\left(r^{m-1}-1\right) \left(r^{nm+n}-1\right) 
- n(m -1)\left(r^{m+1}-1\right) \left(r^{nm+h-1}-1\right)r^{n-h-1}  = 0.
\end{align}
The left-hand side above expands to the polynomial 
\begin{multline}
 (n+h-1)r^{nm+m+n-1}-(nm+h-1)r^{nm+n}+n(m-1)r^{nm+n-2}\\
 +n(m-1)r^{n+m-h}-n(m-1)r^{n-h-1}-(nm+h-1)r^{m-1}+nm+h-1.  \label{eq:SH_transamppoly}
\end{multline}
We will use Descartes' rule of signs \cite{Descartes1637Geometrie} to show that this polynomial has exactly one root greater than 1.
Given $m\ge2$ and $n-2\ge h\ge1$, the terms of the polynomial of Eq.~\eqref{eq:SH_transamppoly} are in descending order of exponent, with the possible exception of the fifth and sixth terms.  Since the quantities $n+h-1,\,nm+h-1,$ and $n(m-1)$ are all positive, there are four sign changes between the consecutive  coefficients of this polynomial. (The coefficients of the fifth and sixth terms are both negative so their ordering does not matter.) By Descartes' rule of signs, this polynomial can have zero, two, or four positive roots, counting multiplicity. From Eq.~\eqref{eq:SH_transamp} we infer that $r=1$ is at least a double root of this polynomial, since both terms of the left-hand side Eq.~\eqref{eq:SH_transamp} contain a factor of the form $(r^x-1)(r^y-1)$ for $x,y \geq 1$.  Furthermore, $r=1$ must be a root of odd multiplicity, since our weak-selection analysis implies a sign change in $ \rho_{SH_{n,m,h}}(r)-\rho_{K_{nm+h}}(r)$ at $r=1$.  We conclude that $r=1$ is a triple root, and that the polynomial of Eq.~\eqref{eq:SH_transamppoly} has exactly one other positive root.  This other root must be greater than 1 since we have shown $ \rho_{SH_{n,m,h}}(r) <\rho_{K_{nm+h}}(r)$ for $0<r<1$.  Therefore, $ \rho_{SH_{n,m,h}}(r) <\rho_{K_{nm+h}}(r)$ only on some connected interval $(1,r^*)$.
\end{case}
Combining the two subcases of Case 3 shows that the Separated Hubs graph is a transient amplifier for $n \ge h+2$. \qedhere
\end{case}
\end{proof}
 
 \subsection{Fan}\label{sect:fan}
The Fan is the $h=1$ case of the Separated Hubs graph; that is, $F_{n,m} = SH_{n,m,1}$.  Substituting $h=1$ into Eqs.~\eqref{eq:SHNenum} and \eqref{eq:SHNedenom} gives the effective population size for the Fan $F_{n,m}$:
\begin{align}
\Ne&=\frac{nm (nm+4 \epsilon -1) (m (nm+\epsilon -1)-\epsilon )}{(m-1) \epsilon ^2 (nm+1)+\epsilon  (m (n+2)-1) (nm-1)+m (nm-1)^2}.
 \end{align}
 For $\epsilon \to 0$, this becomes $\Ne = nm+n-1$, which matches Eq.~\eqref{eq:SHNe}.

For arbitrary mutant fitness $r$, and $\epsilon \to 0$, substituting $h=1$ into Eq.~\eqref{eq:SH_fixprobapp} gives the fixation probability
\begin{equation}
\label{eq:MF_fixprob}
\rho_{F_{n,m}}(r)=\left(\frac{n(m-1)}{nm+1}\right) 
\left(\frac{(1-r^{-1})(1-r^{-(m+1)})}{(1-r^{-(m-1)})(1-r^{-n(m+1)})} \right).
\end{equation}
From Theorem \ref{thm:SHclassify}, the $m$-Fan is a reducer when $n=2$, a suppressor when $n=1$, and a transient amplifier when $n\ge 3$.  

An interesting case is the 2-Fan ($m=2$), with $r>1$ and $n \to \infty$.  In this case, we have
\begin{equation}
    \lim_{n\to\infty}\rho_{F_{n,2}}(r) = \frac{1-r^{-3}}{2}.
\end{equation}

Setting this equal to the large-population limit for a complete graph, $\lim_{N\to \infty}\rho_{K_N}(r)=1-r^{-1}$, allows us to determine the fitness at which the graph is no longer amplifying selection. Solving
\begin{equation}
    \frac{1-r^{-3}}{2} =1-r^{-1},
\end{equation}
we obtain $r=1$ or $r=(1 \pm \sqrt{5})/2$.  Thus the 2-Fan, as the number of blades goes to infinity,   amplifies selection on the interval $1<r\leq\phi$, where $\phi=(1+\sqrt{5})/2$ is the the golden ratio.

\subsection{Star of Islands}
\label{app:SI}
 The Star of Islands graph $SI_{n,m,h}$ consists of $h \geq 2$ ``hub'' vertices and $n\geq 2$ ``islands'', with $m \geq 2$ vertices per island. Within the hub and each island, vertices are connected to one another with weight 1. Each hub vertex is also connected to each island vertex with weight $\epsilon$. 

\subsubsection{Weak Selection}
\label{app:SIweak}
Let $H$ represent a hub vertex and $B$ represent any island vertex. The weighted degree and relative weighted degrees of these vertices are
  \begin{align*}
      w_H&=h-1+nm\epsilon\\
w_B&=m-1+h\epsilon\\
\pi_H&=\frac{w_H}{h w_H+nm w_B}=\frac{h-1+nm\epsilon}{h(h-1+nm\epsilon)+nm(m-1+h\epsilon)}\\
\pi_B&=\frac{w_B}{h w_H+nm w_B}=\frac{m-1+h\epsilon}{h(h-1+nm\epsilon)+nm(m-1+h\epsilon)}.
  \end{align*}  
Furthermore, let $B$ and $B'$ represent any two vertices on the same island and $H$ and $H'$ be represent any two hub vertices. The step probabilities are: 
\begin{align*}
p_{BB'}&=\frac{1}{w_B}=\frac{1}{m-1+h\epsilon}\\
p_{BH}&=\frac{\epsilon}{w_B}=\frac{\epsilon}{m-1+h\epsilon}\\
p_{HB}&=\frac{\epsilon}{w_H}=\frac{\epsilon}{h-1+nm\epsilon}\\
p_{HH'}&=\frac{1}{w_H}=\frac{1}{h-1+nm\epsilon}.
\end{align*}
All other step probabilities are zero.

As for the Separated Hubs graph, there are four coalescence times to consider: $\tau_{HH'}$ for two different hub vertices, $\tau_{HB}$ for a hub and a blade, $\tau_{BB'}$ for two different vertices on a common blade, and $\tau_{BB''}$ for two vertices on different blades.  Eq.~\eqref{eq:trecurapp} becomes  
\begin{align*}
\tau_{HH'}&=1+\frac{1}{2}\big(2nm p_{HB}\tau_{BH}+2(h-2)p_{HH'}\tau_{HH'}\big)\\
\tau_{BH}&=1+\frac{1}{2}\big((m-1)p_{HB}\tau_{BB'}+m(n-1)p_{H,B}\tau_{BB''}+(h-1)p_{HH'}\tau_{BH}+(m-1)p_{BB'}\tau_{BH}\\
&+(h-1)p_{B,H}\tau_{HH'}\big)\\
\tau_{BB'}&=1+\frac{1}{2}\big(2(m-2)p_{BB'}\tau_{BB'}+2hp_{BH}\tau_{BH}\big)\\
\tau_{BB''}&=1+\frac{1}{2}\big(2(m-1)p_{BB'}\tau_{BB''}+2h p_{BH}\tau_{BH}\big).
\end{align*}

Solving this system and substituting into Eqs.~\eqref{eq:tidef} and \eqref{eq:Nedef} yields the effective population size for this graph.  The result for arbitrary $\epsilon$ is rather cumbersome, but in the $\epsilon \to 0$ limit it simplifies to
\begin{equation}
\lim_{\epsilon\rightarrow0}\Ne = N + 
\frac{(m-h)nmh\big(h(h-1)+m(m-1)(n-2) \big)}
{\big(h(h-1)+m(m-1)\big)\big(h(h-1)+m(m-1)n\big)}.
\end{equation}


Fixation probabilities for weak selection can be computed by substituting in Eq.~\eqref{eq:rhoNe}. In the limit of many islands we have
\begin{equation}
\label{eq:SI_nepslimit_weak}
\lim_{n\rightarrow\infty}\lim_{\epsilon\rightarrow0}\rho'=\frac{(m-1)(h+m)}{2(h(h-1)+m(m-1))}.
\end{equation}

To determine the strongest amplifier of weak selection, we observe that Eq.~\eqref{eq:SI_nepslimit_weak} is maximized for $m=h+\sqrt{2h(h-1)}$. Substituting this maximizing value into eq.~\eqref{eq:SI_nepslimit_weak}, we obtain
\begin{align*}
   \max_m \lim_{n\rightarrow\infty}\lim_{\epsilon\rightarrow0}\rho'=
    \frac{-1+h\big(2h-3+2\sqrt{2h(h-1)}\big)}{-2+8h(h-1)}.
\end{align*}
We call the right-hand side $g(h)$. The derivative of $g$ is
\begin{align*}
\frac{dg}{dh}=\frac{1+4h^3-5h+(3-4h^2)\sqrt{2h(h-1)}}{2(h-1)(1+4h-4h^2)^2}.
\end{align*}
Note that since $h\geq2$, the denominator is positive. Moreover, we have $\sqrt{2h(h-1)}\geq h$, and it follows that the numerator of $g'(h)$ is negative:
\begin{align*}
1+4h^3-5h+(3-4h^2)\sqrt{2h(h-1)} \leq 1+4h^3-5h+(3-4h^2)h=1-2h < 0.
\end{align*}
Thus $g(h)$ is decreasing in $h$ and is maximized when $h=2$.  The largest value of $\rho'$ therefore occurs when $h=2$, $m=4$ and $\rho'=\frac{9}{14}$.  For comparison, the complete graph $K_N$ has $\rho'=1/2$ in the $N \to \infty$ limit.

\subsubsection{Nonweak Selection}
\label{app:SInonweak}

Here we compute the fixation probability for the Star of Islands graph in the $\epsilon\to 0$ limit. As in Appendix \ref{sec:SHnonweak}, we use a separation of timescales argument. Here, fixation on the hub or on the islands occurs on a fast timescale, while the spread of types between the hub and the island occurs on a slow timescale.  On the slow timescale, the only relevant states are those for which the hub and islands each contain all mutants or all residents.  We label such states as $(M,k)$ or $(R,k)$, where the first coordinate indicates the type of the hub (M or R), and the second indicate the number $k$ of islands that have been fixed for mutants, $0 \leq k \leq n$.  Dynamics on the slow timescale can be represented as a continuous-time Markov chain.  The transition rate from state $(M,k)$ to $(M, k+1)$ is given by
\begin{equation}
Q_{(M,k) \to (M, k+1)} = \left( \frac{(n-k)m}{nm+h} \right)
\left( \frac{hr\epsilon}{m-1+hr\epsilon} \right)
\left(\frac{m-1}{m} \frac{1-r^{-1}}{1-r^{-(m-1)}}\right).
\end{equation}
The first factor on the right-hand side is the probability that, in a given time-step, a resident on an island vertex dies.  The second factor is the probability that a mutant is chosen to reproduce into the vacant island vertex.  The third factor is the probability that, from this initial mutant invader, the mutant type ultimately achieves fixation on this island. 

Similarly, we define the three other transition rates:
\begin{align}
Q_{(M,k) \to (R, k)} & = \left( \frac{h}{nm+h} \right)
\left( \frac{(n-k)m\epsilon}{(h-1)r+(n-k+kr)m\epsilon} \right)
\left(\frac{h-1}{h} \frac{1-r}{1-r^{h-1}}\right)\\
Q_{(R,k) \to (R, k-1)} &= \left( \frac{km}{nm+h} \right)
\left( \frac{h \epsilon}{(m-1)r+h\epsilon} \right)
\left(\frac{m-1}{m} \frac{1-r}{1-r^{m-1}}\right)\\
Q_{(R,k) \to (M, k)} &= \left( \frac{h}{nm+h} \right)
\left( \frac{mkr \epsilon}{h-1+(n-k+kr)m\epsilon} \right)
\left(\frac{h-1}{h} \frac{1-r^{-1}}{1-r^{-(h-1)}}\right).
\end{align}
All other transitions are impossible, and are therefore assigned rate zero.

We discretize the above Markov chain by defining following transition probabilities, conditioned on leaving the current state, in the $\epsilon \to 0$ limit:
\begin{align}
\nonumber
    P_{(M,k) \to (M, k+1)} & = 
    \lim_{\epsilon\to 0} \frac{Q_{(M,k) \to (M, k+1)}}{Q_{(M,k) \to (M, k+1)}+Q_{(M,k) \to (R, k)}}\\
    &= \frac{hr^m(r^{h-1}-1)}{m(r^{m-1}-1)+ hr^m(r^{h-1}-1)}\\
    \nonumber
    P_{(M,k) \to (R, k)} & = 
    \lim_{\epsilon\to 0} \frac{Q_{(M,k) \to (R, k)}}{Q_{(M,k) \to (M, k+1)}+Q_{(M,k) \to (R, k)}}\\
    &= \frac{m(r^{m-1}-1)}{hr^m(r^{h-1}-1) + m(r^{m-1}-1)}\\
    \nonumber
    P_{(R,k) \to (R, k-1)} & =
    \lim_{\epsilon\to 0} \frac{Q_{(R,k) \to (R, k-1)}}{Q_{(R,k) \to (R, k-1)}+Q_{(R,k) \to (M, k)}}\\
    &=  \frac{h(r^{h-1}-1)}{h(r^{h-1}-1)+mr^h(r^{m-1}-1)}\\
    \nonumber
    P_{(R,k) \to (M, k)} & =
    \lim_{\epsilon\to 0} \frac{Q_{(R,k) \to (M, k)}}{Q_{(R,k) \to (M, k)}+Q_{(R,k) \to (M, k)}}\\
    &= \frac{mr^h(r^{m-1}-1)}{h(r^{h-1}-1)+ mr^h(r^{m-1}-1)}.
\end{align}
We observe that these conditional transition probabilities are independent of $k$.

We now follow the method of analysis that Hadjichrysanthou et al.~\cite{hadjichrysanthou2011evolutionary} developed for the Star graph.  Let $\rho_{(M,k)}$ represent the fixation probability from state $(M,k)$,  and similarly for $\rho_{(R,k)}$.
These fixation probabilities obey the recurrence relations
\begin{align}\label{eq:SI_fixprob_recur}
\begin{split}
\rho_{(M,k)}&= P_{(M,k)\rightarrow(M, k+1)} \, \rho_{(M, k+1)}  + P_{(M,k)\rightarrow(R,k)} \, \rho_{(R,k)}, \,\,0\le k\le n-1,\\
\rho_{(R,k)}&= P_{(R,k)\rightarrow(M,k)} \, \rho_{(M,k)} + P_{(R,k)\rightarrow(R, k-1)} \, \rho_{(R, k-1)},\,\,1\le k\le n,
\end{split}
\end{align}
with boundary conditions $\rho_{(R,0)}=0$ and $\rho_{(M, n)}=1$.

The solution to recurrence relations of the form \eqref{eq:SI_fixprob_recur} was given in Ref.~\cite{hadjichrysanthou2011evolutionary}:
\begin{align*}
\rho_{(M,i)}&= \frac{1+ \sum^{i-1}_{j=1} P_{(M,j)\rightarrow(R,j)}\prod^j_{k=1}\frac{P_{(R,k)\rightarrow(R,k-1)}}{P_{(M,k)\rightarrow(M,k+1)}}}{1+\sum^{n-1}_{j=1}P_{(M,j)\rightarrow(R,j)}\prod^j_{k=1}\frac{P_{(R,k)\rightarrow(R, k-1)}}{P_{(M,k)\rightarrow(M,k+1)}}}\\
\rho_{(R,i)}&= \sum^i_{j=1}P_{(R,j)\rightarrow(M,j)}\,\rho_{
(M,i)}\prod^i_{k=j+1}P_{(R,k)\rightarrow(R,k-1)}.
\end{align*}
In particular, the fixation probabilities starting from a single mutant are:
\begin{align}
\label{eq:rhoR1}
\rho_{(R,1)}&= \frac{P_{(R,1)\rightarrow(M,1)}}{1+\sum^{n-1}_{j=1}P_{(M,j)\rightarrow(R,j)}\prod^j_{k=1}\frac{P_{(R,k)\rightarrow(R,k-1)}}{P_{(M,k)\rightarrow(M,k+1)}}}\\
\label{eq:rhoM0}
\rho_{(M,0)}&= \frac{P_{(M,0)\rightarrow(M,1)}}{1+\sum^{n-1}_{j=1}P_{(M,j)\rightarrow(R,j)}\prod^j_{k=1}\frac{P_{(R,k)\rightarrow(R,k-1)}}{P_{(M,k)\rightarrow(M,k+1)}}}.
\end{align}

Let us define
\begin{equation}
x= \frac{P_{(R,k)\rightarrow(R, k-1)}}{P_{(M,k)\rightarrow(M,k+1)}} = \frac{mr^{-m}(r^{m-1}-1)+h(r^{h-1}-1)}{mr^{h}(r^{m-1}-1)+h(r^{h-1}-1)}.
\end{equation}
Substituting in Eqs.~\eqref{eq:rhoR1} and \eqref{eq:rhoM0}, we obtain
\begin{align}
\rho_{(R,1)}&=\frac{P_{(R,1)\rightarrow(M,1)}}{1+\frac{x^n-x}{x-1}P_{(M,k)\rightarrow(R,k)}}\label{eq:SI_rhoR1}\\
\rho_{(M,0)}&=\frac{P_{(M,0)\rightarrow(M,1)}}{1+\frac{x^n-x}{x-1}P_{(M,k)\rightarrow(R,k)}}.\label{eq:SI_rhoM0}
\end{align}

To obtain the overall fixation probability from a single mutant (in the $\epsilon \to 0$ limit), we must consider two possibilities: (i) the mutation first arises on an island, becomes fixed on that island, and sweeps to fixation from there, or (ii) the mutant first arises on the hub, becomes fixed on the hub, and then sweeps to fixation.  Putting these cases together yields:
\begin{align*}
\rho_{SI_{n,m,h}}(r)=\left( \frac{nm}{nm+h} \right) \left(\frac{m-1}{m} \frac{1-r^{-1}}{1-r^{-(m-1)}}\right)\rho_{(R, 1)}
+\left(\frac{h}{nm+h}\right) \left( \frac{h-1}{h} \frac{1-r^{-1}}{1-r^{-(h-1)}} \right) \rho_{(M, 0)}.
\end{align*}
Substituting from Eqs.~\eqref{eq:SI_rhoR1} and \eqref{eq:SI_rhoM0}, we obtain $\rho_{SI_{n,m,h}}(r)=\text{num/denom}$, where
\begin{multline}
\label{eq:SI_rho1}
    \mathrm{num}=r^m(1-r^{-1})\left(1-r^{-(h+m)}\right)\\
    \times \Big(hr^h\left(1-r^{-(h-1)}\right)\left(nm(m-1)r^m+h(h-1)\right)\\
    +mr^m\left(1-r^{-(m-1)}\right)\left(nm(m-1)+h(h-1)r^h\right)\Big),
\end{multline}
and
\begin{multline}
\label{eq:SI_rho2}
    \mathrm{denom}=(nm+h)\left(h\left(1-r^{-(h-1)}\right) +mr^m\left(1-r^{-(m-1)}\right)\right)\\
    \times\Big(mr^m\left(1-r^{-(m-1)}\right)(1-x^n)+h\left(1-r^{-(h-1)}\right)\left(r^{h+m}-x^n\right)\Big).
\end{multline}

In the limit of many islands, we have
\begin{equation}
\lim_{n \to \infty} \rho_{SI_{n,m,h}}(r) 
= \begin{cases}
0 & 0 \leq r \leq 1\\
\frac{(m-1)(1-r^{-1}) \left(1-r^{-(m-1)} \right)}{hr^{-m}\left(1-r^{-(h-1)}\right) + m\left(1-r^{-(m-1)}\right)}
& r>1.
\end{cases}
\end{equation}

\begin{theorem} Given $m,n>1$ and $\epsilon\rightarrow0$, the Star of Islands is a reducer for $m=h$.
\end{theorem}
\begin{proof}
We wish to compare fixation probabilities on the Star of Islands graph, $SI_{n,m,h}$, to those on a complete graph of equal size, $K_{nm+h}$, in the case $m=h$.  Substituting $h=m$ into \eqref{eq:SI_rho1} and \eqref{eq:SI_rho2}, and comparing to the fixation probability for $K_{nm+m}$ from the main text, we obtain
\begin{equation*}
\frac{\rho_{SI_{n,m,m}}(r)}{\rho_{K_{nm+m}}(r)}
=\frac{\displaystyle \left(\frac{1-r^{-(nm+m-1)}}{nm+m-1}\right)\left(\frac{1-r^{-m}}{m}\right)}
{\displaystyle \left(\frac{1-r^{-(m-1)}}{m-1}\right)\left(\frac{1-r^{-m(n+1)}}{m(n+1)}\right)}.
\end{equation*}
The logarithm of this ratio can be written in terms of the function $F_{r}(x)=\ln\left|\frac{1-r^{-x}}{x}\right|$:
\begin{equation*}
\ln\left(\frac{\rho_{SI_{n,m,m}}(r)}{\rho_{K_{nm+m}}(r)}\right)= F_{r}(nm+m-1)+F_{r}(m)-F_{r}(m-1)-F_{r}(nm+m).
\end{equation*}

From Lemma \ref{lemma_F}, we know that $F_{r}(x)$ is strictly convex in $x$ for all $r\ne 1$. Invoking Lemma \ref{lem:convex} with $a=m-1$, $b=nm+m$, and $d=1$, we see that,
\begin{equation}
F_{r}(m)+F_{r}(nm+m-1)<F_{r}(m-1)+F_{r}(nm+m),
\end{equation}
and therefore $\ln\left(\frac{\rho_{SI_{n,m,m}}(r)}{\rho_{K_{nm+m}}(r)}\right)<0$, for all $r \ne 1$, which proves that Star of Islands graph is a reducer for $m=h$.
\end{proof}

\end{document}